	\documentclass[
		paper=letter,10pt,american,pagesize,%
		twocolumn,headings=small,title=small,DIV=15,
		abstract=true,
		headinclude=true,footinclude=true,footlines=-5,twoside=semi%
	]{scrartcl}
	\def\docclass{my-siam}
	\def\version{arxiv}

	\def\draftmode{false} 

\usepackage[utf8]{inputenc}
\makeatletter
\usepackage{xifthen}

\newcommand\iflipics[2]{\ifthenelse{\equal{\docclass}{lipics}}{#1}{#2}}
\newcommand\ifkoma[2]{\ifthenelse{\equal{\docclass}{koma}}{#1}{#2}}
\newcommand\ifieee[2]{\ifthenelse{\equal{\docclass}{ieee}}{#1}{#2}}
\newcommand\ifsiam[2]{\ifthenelse{\equal{\docclass}{siam}}{#1}{#2}}
\newcommand\ifmysiam[2]{\ifthenelse{\equal{\docclass}{my-siam}}{#1}{#2}}
\newcommand\ifacm[2]{\ifthenelse{\equal{\docclass}{acm}}{#1}{#2}}
\ifthenelse{ \equal{\docclass}{lipics} \OR \equal{\docclass}{koma} \OR \equal{\docclass}{ieee} \OR \equal{\docclass}{siam} \OR \equal{\docclass}{my-siam} \OR \equal{\docclass}{acm} }{
	\PackageInfo{paper}{Building paper with docclass = \docclass} 
}{
	\PackageWarning{paper}{docclass = "\docclass", but must be one of "lipics", "koma", "ieee", "siam", "my-siam", "acm"}
}

\newcommand\ifmanuscript[2]{\ifthenelse{\equal{\version}{manuscript}}{#1}{#2}}
\newcommand\ifarxiv[2]{\ifthenelse{\equal{\version}{arxiv}}{#1}{#2}}
\newcommand\ifsubmission[2]{\ifthenelse{\equal{\version}{submission}}{#1}{#2}}
\newcommand\ifproceedings[2]{\ifthenelse{\equal{\version}{proceedings}}{#1}{#2}}
\ifthenelse{ 
	\equal{\version}{manuscript} 
	\OR \equal{\version}{arxiv} 
	\OR \equal{\version}{submission} 
	\OR \equal{\version}{proceedings} 
}{
	\PackageInfo{paper}{Building paper version = \version} 
}{
	\PackageWarning{paper}{version = "\version", but must be one of "manuscript", "arxiv", "submission", "proceedings"}
}

\newcommand\ifdraft[2]{\ifthenelse{\equal{\draftmode}{true}}{#1}{#2}}
\ifthenelse{ \equal{\draftmode}{true} \OR \equal{\draftmode}{false} }{
	\PackageInfo{paper}{Building paper with draftmode = \draftmode} 
}{
	\PackageWarning{paper}{draftmode = "\draftmode", but must be "true" or "false"}
}

\usepackage[T1]{fontenc}
\ifsiam{
	\usepackage{lmodern}
	\usepackage{slantsc}
}{}
\ifmysiam{
	\usepackage{lmodern}
	\usepackage{slantsc}
}{}
\ifkoma{
	\usepackage{lmodern}
	\usepackage{slantsc}
}{}
\iflipics{
	\usepackage{lmodern}
	\usepackage{slantsc}
}{}

\usepackage{babel}
\input{ushyphex.tex} 

\usepackage{array,multicol}
\ifieee{
	\usepackage[cmex10]{amsmath,mathtools}
	\usepackage{amsfonts,amssymb}
}{}
\ifkoma{
	\usepackage{amsmath,amsfonts,amssymb,mathtools}
}{}
\iflipics{
	\usepackage{amsmath,amsfonts,amssymb,mathtools}
}{}
\ifsiam{
	\usepackage{amsmath,amsfonts,amssymb,mathtools}
}{}
\ifmysiam{
	\usepackage{amsmath,amsfonts,amssymb,mathtools}
}{}
\ifacm{
	\usepackage{mathtools}
}{}

\usepackage{mleftright}\mleftright 
\usepackage{relsize,xspace,booktabs,adjustbox,needspace,pbox,relsize}
\ifieee{
	
	\usepackage{enumitem}
}{
	\usepackage{enumitem}
}
\usepackage{graphicx}

\ifacm{}{
	\usepackage{colonequals}
}

\usepackage{wref}
\usepackage[bibtex]{url-doi-arxiv}

\ifsiam{
	\usepackage{ltexpprt}
}{}

\newdimen\makeboxdimen
\newcommand\makeboxlike[3][l]{%
\setbox0=\hbox{#2}%
\global\makeboxdimen=\wd0%
\setbox1=\hbox{\makebox[\makeboxdimen][#1]{%
\makebox[0pt][#1]{#3}%
}}%
\ht1=\ht0%
\dp1=\dp0%
\box1%
}

\newcommand\like[3][c]{%
	\mathchoice{
		\makeboxlike[#1]{%
			\ensuremath{\displaystyle\relax#2}%
		}{%
			\ensuremath{\displaystyle\relax#3}%
		}%
	}{
		\makeboxlike[#1]{%
			\ensuremath{\textstyle\relax#2}%
		}{%
			\ensuremath{\textstyle\relax#3}%
		}%
	}{
		\makeboxlike[#1]{%
			\ensuremath{\scriptstyle\relax#2}%
		}{%
			\ensuremath{\scriptstyle\relax#3}%
		}%
	}{
		\makeboxlike[#1]{%
			\ensuremath{\scriptscriptstyle\relax#2}%
		}{%
			\ensuremath{\scriptscriptstyle\relax#3}%
		}%
	}
}

\ifthenelse{\equal{\docclass}{koma} \OR \equal{\docclass}{my-siam}}{
	\setlength\parindent{1.5em}
	\usepackage[headsepline]{scrlayer-scrpage}
	\pagestyle{scrheadings}
	\clearscrheadfoot
	\AtBeginDocument{%
		\automark[section]{}%
	}
	\ohead{\pagemark}
	\rehead{\mytitle}
	\lohead{\headmark}
	\addtokomafont{caption}{\sffamily\small}
	\addtokomafont{captionlabel}{\sffamily\textbf}
	\setcapmargin{2em}
}{}
\ifthenelse{\equal{\docclass}{my-siam}}{
	\setcapmargin{1em}
	\setcapindent{0em}
}{}
\ifmysiam{
	\setlength\parskip{0pt}
	\RedeclareSectionCommand[
		beforeskip=-1.25\baselineskip,
		afterskip=0.75\baselineskip,
	]{section}
	\RedeclareSectionCommand[
		beforeskip=-1\baselineskip,
		afterskip=-1.5em,
	]{subsection}
	\RedeclareSectionCommand[
		beforeskip=-1\baselineskip,
		afterskip=-1.5em,
	]{subsubsection}
	\RedeclareSectionCommand[
		beforeskip=-.25\baselineskip,
		indent=1.5em,
		afterskip=-1em,
	]{paragraph}
}{}

\AtBeginDocument{%
	\let\mytitle\@title%
}

\let\oldthebibliography\thebibliography
\renewcommand\thebibliography[1]{%
	\oldthebibliography{#1}%
	\pdfbookmark[1]{References}{}%
}

\usepackage{lscape} 

\ifkoma{
	\usepackage{float}
	\floatstyle{plain}
	\usepackage{newfloat}
	\DeclareFloatingEnvironment[%
			name=Algorithm,%
			placement=thb,%
		]{algorithm}
}{}
\iflipics{
	\usepackage{newfloat}
	\DeclareFloatingEnvironment[%
			name=Algorithm,%
			placement=thb,%
		]{algorithm}
}{}
\ifmysiam{
	\usepackage{newfloat}
	\DeclareFloatingEnvironment[%
			name=Algorithm,%
			placement=thb,%
		]{algorithm}
}

\ifmysiam{
	\setcounter{topnumber}{3}
	\setcounter{bottomnumber}{3}
	\setcounter{totalnumber}{3}     
	\setcounter{dbltopnumber}{3}    

}{}
\ifsiam{

	\setcounter{topnumber}{3}
	\setcounter{bottomnumber}{3}
	\setcounter{totalnumber}{3}     
	\setcounter{dbltopnumber}{3}    

}{}

\usepackage{dcolumn}

\usepackage{wclrscode}

\usepackage{textcomp} 
\usepackage{listings}

\usepackage{tikz}

\usetikzlibrary{positioning,arrows.meta,fit}
\usetikzlibrary{backgrounds,calc,trees,graphs}

\pgfdeclarelayer{background}
\pgfsetlayers{background,main}

\iflipics{
	\newtheorem{fact}[theorem]{Fact}
	\newcommand\thmemph[1]{\emph{#1}}
	\newenvironment{proofof}[1]{%
		\begin{proof}[{{Proof of #1{}}}]%
	}{%
		\end{proof}%
	}
}{}
\ifacm{
	\AtEndPreamble{
		\theoremstyle{acmdefinition}
		\newtheorem{remark}[theorem]{Remark}
		\newtheorem{fact}[theorem]{Fact}
	}
	\newcommand\thmemph[1]{\emph{#1}}
	\newenvironment{proofof}[1]{%
		\begin{proof}[{{Proof of #1{}}}]%
	}{%
		\end{proof}%
	}
}{}
\ifsiam{
	\newtheorem{remark}{Remark}
	\newenvironment{proofof}[1]{%
			\begin{proof}[{{#1{}}}]%
		}{%
			\end{proof}%
		}
}{}
\ifthenelse{\equal{\docclass}{lipics} \OR \equal{\docclass}{siam} \OR \equal{\docclass}{acm}}{}{
	\usepackage[amsmath,hyperref,thmmarks]{ntheorem}
	
	\theorembodyfont{\slshape}
	\theoremseparator{:}
	\newtheoremstyle{proofstyle}%
	  {\item[\theorem@headerfont\hskip\labelsep ##1\theorem@separator]}%
	  {\item[\theorem@headerfont\hskip\labelsep ##3\theorem@separator]}
	
	\theorempreskip{\topsep} 

	\theoremsymbol{\adjustbox{scale=.8}{$\triangleleft\mkern-1mu$}}
	
	\newtheorem{theorem}{Theorem}[section]
	
	\theoremstyle{plain}
	\theorempreskip{\topsep}
	
	\newtheorem{proposition}[theorem]{Proposition}

	\newtheorem{definition}[theorem]{Definition}
	
	\theoremstyle{plain}
	\theorembodyfont{\upshape}

	\newtheorem{fact}[theorem]{Fact}
	\newtheorem{remark}[theorem]{Remark}

	\theoremsymbol{\raisebox{-.25ex}{$\Box$}}
	\qedsymbol{\raisebox{-.25ex}{$\Box$}}
	
	\theoremstyle{proofstyle}
	\newtheorem{proof}{Proof}

	\newcommand\thmemph[1]{\textit{#1}}
}

\iflipics{
	\newenvironment{thmenumerate}[2][]{%
		\begin{enumerate}[
			label={\textsf{\textbf{\color{darkgray}{\makebox[\widthof{(a)}][c]{\textup{(\alph*)}}}}}},
			ref={\ref{#2}\kern.1em--\kern.1em(\alph*)},
			itemsep=0pt,
			topsep=.5ex,
			leftmargin=1.75em,
			#1
		]%
	}{%
		\end{enumerate}%
	}
}{
	
}

\newcommand*\ie{\mbox{i.\hspace{.2ex}e.}}
\newcommand*\eg{\mbox{e.\hspace{.2ex}g.}}

\newcommand\N{\mathbb N}

\usepackage{fixmath}

\newcommand{\ESymbol}{\mathbb{E}}

\newcommand{\ProbSymbol}{\ensuremath{\mathbb{P}}}

\DeclarePairedDelimiterXPP\Prob[1]{\ProbSymbol}[]{}{%
	#1%
}
\DeclarePairedDelimiterXPP\E[1]{\ESymbol}[]{}{%
	#1%
}
\DeclarePairedDelimiterXPP\Eover[2]{\ESymbol_{#1}}[]{}{%
	#2%
}
\DeclarePairedDelimiterXPP\ProbIn[2]{\ProbSymbol_{#1}}[]{}{%
	#2%
}
\providecommand{\Prob}{} 
\providecommand{\ProbIn}{} 
\providecommand{\E}{} 
\providecommand{\Eover}{} 

\newcommand{\surroundedmath}[3]{
	\mathchoice{
		#1{#2{#3}#2}%
	}{
		#1{#3}%
	}{
		#1{#3}%
	}{
		#1{#3}%
	}%
}
\newcommand\rel[1]{\surroundedmath{\mathrel}{\:}{#1}}

\iflipics{}{
	\makeatletter
	\let\oldalign\align
	\let\endoldalign\endalign
	\renewenvironment{align}{%
		\begingroup%
		\let\oldhalign\halign
		\def\halign{%
			\let\oldbreak\\%
			\def\nonnumberbreak{\nonumber\oldbreak*}%
			\def\\{%
				\@ifstar{\nonnumberbreak}{\oldbreak}%
			}%
			\oldhalign%
		}
		\oldalign%
	}{%
		\endoldalign%
		\endgroup%
	}
}
\newcommand*\numberthis[1][]{\stepcounter{equation}\tag{\theequation}}

\allowdisplaybreaks[3]

\newcommand\splitaftercomma[1]{%
  \begingroup
  \begingroup\lccode`~=`, \lowercase{\endgroup
    \edef~{\mathchar\the\mathcode`, \penalty0 \noexpand\hspace{0pt plus .25em}}%
  }\mathcode`,="8000 #1%
  \endgroup
}

\usepackage{soul} 

\def\mydots{\xleaders\hbox to.5em{\hfill.\hfill}\hfill}
\newlength\tmpLenNotations

\ifdraft{%
    \usepackage[switch]{lineno} 
    \linenumbers
	\overfullrule=6mm
	
	\usepackage[color,notref,notcite]{showkeys}
	\definecolor{refkey}{gray}{.99}
	\colorlet{labelkey}{green!60!black!60}
	
	\usepackage[inline,nolabel]{showlabels}

	\showlabels{cite}
	\showlabels{citealt}
	\showlabels{citealp}
	\showlabels{citet}
	\showlabels{Citet}
	\showlabels{citep}
	\showlabels{citeauthor}
	\showlabels{Citeauthor}
	\showlabels{citefullauthor}
	\showlabels{citeyear}
	\showlabels{citeyearpar}
	\showlabels{wref}
	\showlabels{wpref}
	\showlabels{wtpref}
	\showlabels{wildref}
	\showlabels{wildpageref}
	\showlabels{wildtpageref}
}{}

\iflipics{
	\ifmanuscript{\hideLIPIcs}{}
	\ifarxiv{\hideLIPIcs}{}
	\ifsubmission{}{\nolinenumbers}
}{}

\ifdraft{}{%
	\usepackage{microtype}
}

\hypersetup{
	final,
	unicode=true, 
	bookmarks=true,
	bookmarksnumbered=true,
	bookmarksdepth=2,
	bookmarksopen=true,
	breaklinks=true,
	hidelinks,
}

\newsavebox\tmpbox

\iflipics{
	\let\oldparagraph\paragraph
	\renewcommand\paragraph[1]{%
		\oldparagraph*{#1}
	}
}{
	\let\oldparagraph\paragraph
	\renewcommand\paragraph[1]{%
		\oldparagraph{#1.}
	}
}

\ifmysiam{
	\let\oldsubsection\subsection
	\renewcommand\subsection[1]{%
		\oldsubsection{#1.}%
	}
	\let\oldsubsubsection\subsubsection
	\renewcommand\subsubsection[1]{%
		\oldsubsubsection{#1.}%
	}
}{}
\ifsiam{
	\let\oldsubsection\subsection
	\renewcommand\subsection[1]{%
		\oldsubsection{#1.}%
	}
	\let\oldsubsubsection\subsubsection
	\renewcommand\subsubsection[1]{%
		\oldsubsubsection{#1.}%
	}
}{}

\let\epsilon\varepsilon

\def\myacknowledgements{}
\ifkoma{
	
}{}
\ifieee{
	
}{}
\ifsiam{
	
}{}
\ifmysiam{
	
}{}
\ifacm{
	
}{}

\newcommand\leafnode[1]{%
	\ensuremath{R_{#1}}\xspace%
}

\newcommand\Cpp{C\texttt{++}\xspace}

\newcommand\runboundary[1]{%
	\ensuremath{B_{#1}}\xspace%
}

\makeatother

\usepackage{hyperref}
\ifacm{
	\title{Multiway Powersort}
}{
	\title{Multiway Powersort}
}

\ifthenelse{\NOT \equal{\docclass}{lipics} \AND \NOT \equal{\docclass}{acm}}{
	\newcommand\email[1]{\texttt{#1}}
	\ifsubmission{%
		\author{Anonymous Authors}
	}{%
	\author{%
		William Cawley Gelling%
			\footnote{\strut University of Liverpool, UK, \protect\\
			\email{William.cawleygelling\,@\,gmail.com}}
	\and 
		Markus E.\ Nebel%
			\footnote{Bielefeld University, Germany, 
			\email{nebel\,@\,tf.unibi.de}}
	\and 
		Benjamin Smith%
			\footnote{University of Liverpool, UK, 
			\email{b.m.smith\,@\,liverpool.ac.uk}}
	\and
		Sebastian Wild%
			\footnote{University of Liverpool, UK, \protect\\
			\email{sebastian.wild\,@\,liverpool.ac.uk}}
	}
	}
	
	\date{\small\today}
}{}

\begin{document}

\ifacm{}{\maketitle} 

\begin{abstract}
We present a stable mergesort variant, {\em Multiway Powersort},
that exploits existing runs and finds nearly-optimal merging orders for \emph{$k$-way} merges
with negligible overhead.
This builds on Powersort (Munro\,\&\,Wild, ESA 2018), which has recently replaced Timsort's suboptimal merge policy in the CPython reference implementation of Python, as well as in PyPy and further libraries.
Multiway Powersort reduces the number of memory transfers, which increasingly determine the cost of internal sorting (as observed with Multiway Quicksort (Kushagra et al., ALENEX 2014; Aumüller\,\&\,Dietzfelbinger, TALG 2016; Wild, PhD thesis 2016) and the inclusion of Dual-Pivot Quicksort in the Java runtime library).
We demonstrate that our 4-way Powersort implementation can achieve substantial speedups 
over standard (2-way) Powersort and other stable sorting methods without compromising the optimally run-adaptive performance of Powersort.
\end{abstract}

\ifacm{%
	\maketitle%
}{}

\section{Introduction}

Sorting a list of items is one of the foundational problems of computer science. Prominent
sorting algorithms include Quicksort and Mergesort, for which highly engineered
implementations exist. While Quicksort is typically faster, Mergesort has
the optimal $\Theta(n\log(n))$ worst-case behavior and is stable, \ie, 
it preserves the relative order of elements that compare equal under the sorting criterion.
Stability is a central requirement in many programming libraries, \eg, 
the Java runtime library requires the sorting method for objects to be stable.
Indeed, the need of a fast and stable general-purpose sorting method for the 
CPython reference implementation of the Python programming language was 
the main motivation for Tim Peters to develop a new variant of Mergesort, known as Timsort~\cite{Peters2002mailinglist}.

Apart from its stability, Mergesort is also particularly well-suited for \emph{adaptive sorting},
\ie, for saving effort when the input is already partially sorted (see also \wref{sec:preliminaries}).
Again, Timsort did pioneering work in bringing such adaptive sorting to most modern standard libraries
(including Python, Java, Android runtimes, the V8 Javascript engine, Rust, Swift, Apache Spark, Octave, and the NCBI \Cpp Toolkit),
giving users linear complexity for sufficiently sorted inputs.

More specifically, Timsort takes advantage of existing \emph{runs}, \ie,
contiguous regions of the input that are already in sorted order.
Timsort detects such runs on the fly and maintains a stack of runs yet to be merged. 
When the next run is found, a particular merge policy triggers a (potentially empty) sequence of merges 
before that new run is added to the stack;
this merge policy determines which merges are executed. 
Since the runs in the input can vary wildly in length,
the chosen merge order can have a huge impact on efficiency.
Consider as an indicative example an input with one run of length $n/2$ and $n/4$ runs of size 2.
Repeatedly merging the long run with one small run results would result in $\Theta(n^2)$ overall time,
whereas first merging all small runs (using a balanced tree) into a single run and then combined
this run with the original long run is much cheaper with $\Theta(n \log n)$ time.

Despite its vast success in practice, 
increased scrutiny revealed several issues with the original merge policy of Timsort:
First, an algorithmic bug allowed the run stack to grow larger than anticipated, \emph{twice} leaving widely deployed libraries vulnerable to memory access violations upon sorting an adversarial input~\cite{DeGouwBoerBubelHaehnleRotSteinhoefel2017,AugerJugeNicaudPivoteau2018}.
Second, Timsort incurs up to 50\% overhead over an optimal merge policy for certain inputs~\cite{BussKnop2019,AugerJugeNicaudPivoteau2018,AugerJugeNicaudPivoteau2018arxiv} 
with noticeable effects on running times~\cite{MunroWild2018}.
These deficiencies have sparked a flurry of work on suggesting possible 
alternatives~\cite{AugerJugeNicaudPivoteau2018,MunroWild2018,BussKnop2019,Juge2020}, some with much stronger efficiency guarantees.
After proving efficient in running time studies~\cite{PetersEtAl2018}, 
Powersort~\cite{MunroWild2018} has now replaced Timsort's original merge policy in CPython, PyPy, and AssemblyScript~\cite{Peters2021}.

Multiway algorithms have also seen recent adoption, with a surprising change in the Java runtime library's (unstable) method for sorting primitive-typed arrays.
There, a state-of-the-art implementation of classic
Quicksort~-- the standard sort for decades~-- was replaced in 2011 by a novel dual-pivot variant 
after the latter proved 10\% faster in practice. 
Detailed analysis revealed that the observed speedup is due to memory transfers: dual-pivot Quicksort
profits from savings in \emph{``scanned elements''}, the number of touched memory cells not in cache~\cite{KushagraLopezOrtizQiaoMunro2013,AumullerDietzfelbinger2016,NebelWildMartinez2016}.
Since improvements in CPU speed have outpaced improvements in memory bandwidth over several decades,
it now pays to avoid memory transfers in internal sorting, even at the expense of slightly increasing effort inside the CPU or more complicated code~\cite{AumullerDietzfelbinger2016,Wild2016}.
Multiway partitioning in Quicksort resp.\ multiway merging in Mergesort does exactly that: 
By getting more of the sorting task done in one pass over the data, the overall data volume transferred between memory and CPU is reduced. 
An abstract cost measure called {\em scanned elements} has proven appropriate to capture corresponding effects.

While multiway merging is an elementary trick of the trade to reduce the memory-transfer cost in Mergesort,
applying it to state-of-the-art adaptive sorting methods is a delicate task:
to not forsake the gains on partially sorted inputs, 
we have to generalize the policy for deciding which runs to merge.
Indeed in the case of Timsort, it seems quite unclear how to utilize multiway merges; Timsort's local rules are tailored to 2-way merges. Apart from Peeksort and Powersort, the same is true for all other practical methods suggested to replace Timsort, in particular adaptive Shivers Sort~\cite{Juge2020} and $\alpha$-MergeSort~\cite{BussKnop2019}.
By contrast, Powersort is based on a principled construction, 
mimicking nearly-optimal binary search tree algorithms~-- which allows us to present a multiway generalization.

\subsection{Contributions}
In this paper, we show how to use multiway merging in nearly-optimal Mergesort variants and that significant performance improvements result from that.
More precisely, our contributions are as follows:

(1) We design $k$-way Powersort, a provably optimal generalization of Powersort which simultaniously merges $k$ runs wherever this is possible.
To our knowledge, this is the first run-adaptive multiway Mergesort variant presented in the literature.
We prove that~-- just like standard 2-way Powersort~-- $k$-way Powersort uses $\le n\cdot \mathcal H + 3n$ comparisons to sort a list of $n$ elements, where $\mathcal H$ is the run-length entropy (defined in \wref{sec:preliminaries}) when $k$ is a power of two;
this is optimal in the worst case.
We furthermore show that the asymptotic worst-case merge cost for $k$-way Powersort is a factor $\log_2(k)$ smaller than for 2-way Powersort. Here the merge cost counts the total number of output elements of all merges conducted by an algorithm and thus approximates the amount of memory transfers in Mergesort variants. 

(2) We engineer an efficient implementation of 4-way Powersort in \Cpp
and demonstrate in an extensive running time study that it realizes performance improvements between
15 \% and 20\% over standard 2-way Powersort. This observation is robust against changes of the size of elements to be sorted (we compare integer and record data) and different amounts of presortedness of the input. 
Besides wall-clock running times, we determine scanned elements and cache performance in valgrind/cachegrind to explain the observed speedup.

(3) We confirm that the results reported in~\cite{MunroWild2018} for Java implementations of the algorithms
are reproducible in a low-level implementation.

\subsection{Related Work}
\label{sec:related-work}

Adaptive sorting refers to any algorithms that exploit structure in the input to improve their performance.
A survey of adaptive sorting up to 2018 and the story behind Timsort was given in~\cite{MunroWild2018},
so we focus on developments since then.
After~\cite{BussKnop2019} pointed out the suboptimal worst-case of Timsort's merge policy,
several papers suggested alternatives~\cite{AugerJugeNicaudPivoteau2018,Juge2020,MunroWild2018}.
The most noteworthy of these is Adaptive Shivers Sort~\cite{Juge2020}, which shares many favorable properties with Powersort and proves the existence of a merge policy that is optimal up to an $O(n)$ term in the theoretical model of Buss and Knop~\cite{BussKnop2019}. It is not clear whether Shivers Sort can easily be generalized to multiway merging.

A second related work studies Timsort's galloping merge method in detail.
Ghasemi et al.~\cite{GhasemiJugeKhalighinejad2022} show that most Mergesort variants can sort in $\mathcal H^* n + O(n)$ comparisons when using a galloping merge, where $\mathcal H^*$ is the ``dual-run-length entropy'' (which is at most the classical entropy of multiplicities of equal elements).
Note that for sorting arrays, galloping only saves on comparisons; one still has to move elements around.
The improvements obtainable from galloping merges are thus largely orthogonal to our work.

\section{Preliminaries}
\label{sec:preliminaries}

If considering a bottom-up approach for Mergesort, one can take advantage of continuous regions of the input that are already
sorted. We call such a region a \emph{run} and define both maximal weakly increasing (\ie, non-decreasing), or strictly decreasing
regions as runs~-- the latter can easily be reversed without issues towards stability. In the sequel, we will always denote the number 
of runs in the input by $r$ and by $L_0,\ldots, L_{r-1}$ their respective lengths such that $L_0+\cdots+L_{r-1}=n$ holds. 
From an algorithmic point of view, the unique partitioning of the input into sorted segments (runs) can be found by an initial scan of the input.

\subsection{Lower Bound}
In lower bounds for the cost of sorting, the binary Shannon entropy ${\mathcal H}$ often shows up:
For $p_1,\ldots,p_m \in (0,1)$ with $p_1 + \cdots + p_m = 1$ we define ${\mathcal H}(p_1,\ldots,p_m) = \sum p_i \lg(1/p_i)$, 
where here and throughout 5$\lg = \log_2$.

\begin{theorem}[{{\cite[Thm.\,2]{BarbayNavarro2013}}}]
	Sorting a list of $n$ elements with initial runs of lengths $L_0,\ldots, L_{r-1}$
	requires $\mathcal H(\frac{L_0}n,\ldots,\frac{L_{r-1}}n) n - O(n)$ 
	comparisons in the worst case.
\end{theorem}

Where the run lengths are clear from the context, we will use $\mathcal H$ to abbreviate $\mathcal H(\frac{L_0}n,\ldots,\frac{L_{r-1}}n)$.

\subsection{Merging and its memory-transfer cost}

Merging describes the process of combining two or more sorted lists (\ie, runs) into a single, sorted list
containing the same elements, by means of pairwise comparisons of elements.
All our merge methods retain the input runs in a buffer area and produce the merged output in another area.
(In-place methods exist, but their time overhead renders them uncompetitive for our purposes.)
Important metrics for merging methods include the number of key comparisons and the number of \emph{scanned elements,} \ie,  memory accesses (read or write) that are not served from cache~\cite{NebelWildMartinez2016}.

Merging can be achieved by iteratively finding the minimum among the smallest remaining elements of the runs
and moving it to the output. 
When merging $k$ input runs at once, maintaining a \emph{tournament tree}~\cite{Knuth1998} of the run minima
(initialized using $k-1$ comparisons) allows us to find the next output element after $\lceil \lg(k)\rceil$ or $\lfloor \lg(k)\rfloor$ comparisons. So we can merge $k$ runs with $n$ elements in total using at most $\lceil \lg(k)\rceil n + k-1$ comparisons; unless the input runs are of very different lengths, this strategy is essentially optimal (for worst-case comparison cost).
When the input runs already reside in the buffer area, we can assume that each element is ``scanned'' exactly twice (fetched once and output once); it can be read and compared many times, but these accesses can be assumed to be cached (for moderate $k$).

\subsection{Run-adaptive Mergesort}
\label{sec:run-adaptive-mergesort}

Mergesort operates by successively merging runs until only a single one remains.
To get started, classic top-down Mergesort considers each individual element as a trivial run.
An \emph{adaptive} (\ie, \emph{run-adaptive}) Mergesort instead finds (contiguous) segments of the input that are already in order.

Since both scanned-element count and comparison count are linear in the output size of a merge (for fixed~$k$ and sufficiently long runs), a convenient abstract measure to compare differentMergesort variants is the \emph{merge cost}~\cite{BussKnop2019}, defined as
the total size of the outputs of all the merges.

Let $\leafnode 0,\ldots, \leafnode{r-1}$ be the runs in the input; recall that $L_i$ denotes the length of $R_i$.
Any stable run-adaptive mergesort can be described as a \emph{merge tree}~$T$: 
an (ordered rooted) tree with leaves $\leafnode0,\ldots,\leafnode{r-1}$ (appearing in that order from left to right).
Each internal node $v$ of $T$ corresponds to the result of merging its children;
we write $M(v)$ for the total length of the result of that merge.
The total merge cost of $T$ is then $M=M(T)=\sum_{v \in T} M(v)$.
$T$ is a $k$-way merge tree if all its nodes have degree at most $k$.
We remark here that it is in general not possible to have all merge-tree nodes combine exactly $k$ runs, therefore a $k$-way merge tree is allowed to contain nodes of smaller degree.
(Though we do not allow unary nodes.)

We further denote by $\runboundary i$ the boundary between the $(i-1)$st and $i$th run ($i=1,\ldots,r-1$).
In binary merge trees ($k=2$), there is a one-to-one correspondence between the $\runboundary i$ and 
the (inner/non-leaf) merge-tree nodes; for multiway merging, this becomes many-to-one: each merge-tree node $v$ of degree $d$ is associated to exactly $d-1$ (not necessarily adjacent) run boundaries $B(v,1),\ldots,B(v,d-1)$;
conversely each run boundary $\runboundary i$ is assigned exactly to one merge-tree node $m(\runboundary i)$.

Let $d_i$ be the \emph{depth} of leaf \leafnode i in $T$,
where depth is the number of edges on the path to the root.
Every element in run \leafnode i is counted exactly $d_i$ times in $\sum_{v} M(v)$,
so we have $M = \sum_{i=0}^{r-1} d_i\cdot L_i$.

\section{Multiway Powersort}
\label{sec:powersort}

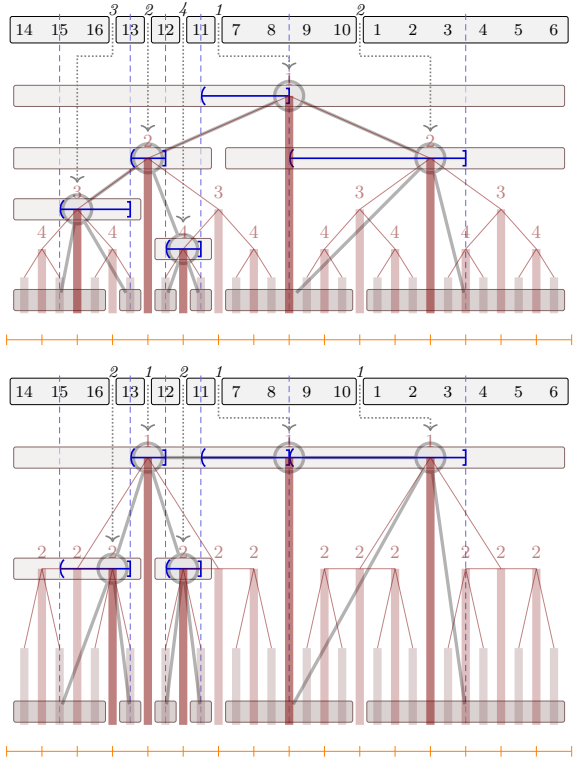
\begin{figure}[tbhp]
	\centering
	\adjustbox{max width=.9\linewidth}{%
		\begin{tikzpicture}[
		every node/.style = {font=\small},
		array entry/.style = {font=\footnotesize},
		power label/.style = {font=\footnotesize\itshape},
		run/.style = {draw=black,fill=black!5,rounded corners=1pt},
		merge result/.style = {run,draw=black!80!red,fill=black!80!red!10,opacity=.6},
		run leaf/.style = {run,draw=black!80!red,fill=black!80!red!30,opacity=.6},
		merge node/.style = {ultra thick,draw=black,fill=black!20,circle,opacity=.3,minimum size=14pt,},
		mid/.style = {densely dashed,blue!80!black!60},
		mid interval/.style = {{Arc Barb[arc=140,width=10pt]}-{Bracket[width=8pt]},shorten >=-.4pt,thick,blue!80!black},
		comb/.style = {line join=round,red!50!black,line width=4pt,opacity=.5},
		comb leaf/.style = {line join=round,red!50!black!40!gray,line width=4pt,opacity=.5},
		xscale=.6,yscale=.45,
]


	\begin{scope}[shift={(0,2)}]
		\draw[run] (0-.4,0) rectangle (2+.4,1) ;
		\draw[run] (3-.4,0) rectangle (3+.4,1) ;
		\draw[run] (4-.4,0) rectangle (4+.4,1) ;
		\draw[run] (5-.4,0) rectangle (5+.4,1) ;
		\draw[run] (6-.4,0) rectangle (9+.4,1) ;
		\draw[run] (10-.4,0) rectangle (15+.4,1) ;
		\draw[mid] (1.000000,1.1) -- ++(0,-8.214286-3.1) ;
		\draw[mid] (3.000000,1.1) -- ++(0,-8.214286-3.1) ;
		\draw[mid] (4.000000,1.1) -- ++(0,-8.214286-3.1) ;
		\draw[mid] (5.000000,1.1) -- ++(0,-8.214286-3.1) ;
		\draw[mid] (7.500000,1.1) -- ++(0,-8.214286-3.1) ;
		\draw[mid] (12.500000,1.1) -- ++(0,-8.214286-3.1) ;
		\node[power label] (B0) at (2+0.5,1.25) {3} ;
		\node[power label] (B1) at (3+0.5,1.25) {2} ;
		\node[power label] (B2) at (4+0.5,1.25) {4} ;
		\node[power label] (B3) at (5+0.5,1.25) {1} ;
		\node[power label] (B4) at (9+0.5,1.25) {2} ;
		\node[array entry] at (0,0.5) {14} ;
		\node[array entry] at (1,0.5) {15} ;
		\node[array entry] at (2,0.5) {16} ;
		\node[array entry] at (3,0.5) {13} ;
		\node[array entry] at (4,0.5) {12} ;
		\node[array entry] at (5,0.5) {11} ;
		\node[array entry] at (6,0.5) {7} ;
		\node[array entry] at (7,0.5) {8} ;
		\node[array entry] at (8,0.5) {9} ;
		\node[array entry] at (9,0.5) {10} ;
		\node[array entry] at (10,0.5) {1} ;
		\node[array entry] at (11,0.5) {2} ;
		\node[array entry] at (12,0.5) {3} ;
		\node[array entry] at (13,0.5) {4} ;
		\node[array entry] at (14,0.5) {5} ;
		\node[array entry] at (15,0.5) {6} ;
	\end{scope}

	\begin{scope}[shift={(0,-8.214286)}]
		\draw[merge result] (0-0.3,8.214286-0.4) rectangle ++(16-0.4,0.8) ;
		\draw[merge result] (0-0.3,5.857143-0.4) rectangle ++(6-0.4,0.8) ;
		\draw[merge result] (0-0.3,3.928571-0.4) rectangle ++(4-0.4,0.8) ;
		\draw[run leaf] (0-0.3,0.500000-0.4) rectangle ++(3-0.4,0.8) ;
		\draw[run leaf] (3-0.3,0.500000-0.4) rectangle ++(1-0.4,0.8) ;
		\draw[merge result] (4-0.3,2.428571-0.4) rectangle ++(2-0.4,0.8) ;
		\draw[run leaf] (4-0.3,0.500000-0.4) rectangle ++(1-0.4,0.8) ;
		\draw[run leaf] (5-0.3,0.500000-0.4) rectangle ++(1-0.4,0.8) ;
		\draw[merge result] (6-0.3,5.857143-0.4) rectangle ++(10-0.4,0.8) ;
		\draw[run leaf] (6-0.3,0.500000-0.4) rectangle ++(4-0.4,0.8) ;
		\draw[run leaf] (10-0.3,0.500000-0.4) rectangle ++(6-0.4,0.8) ;
		\node[merge node] (merge-3) at (7.500000,8.214286) {} ;
		\node[merge node] (merge-1) at (3.500000,5.857143) {} ;
		\node[merge node] (merge-0) at (1.500000,3.928571) {} ;
		\node[] (merge-0l) at (1.000000,0.500000) {} ;
		\node[] (merge-1l) at (3.000000,0.500000) {} ;
		\node[merge node] (merge-2) at (4.500000,2.428571) {} ;
		\node[] (merge-2l) at (4.000000,0.500000) {} ;
		\node[] (merge-3l) at (5.000000,0.500000) {} ;
		\node[merge node] (merge-4) at (11.500000,5.857143) {} ;
		\node[] (merge-4l) at (7.500000,0.500000) {} ;
		\node[] (merge-5l) at (12.500000,0.500000) {} ;
		\draw[merge node] (merge-0) to (merge-0l) ; 
		\draw[merge node] (merge-0) to (merge-1l) ; 
		\draw[merge node] (merge-1) to (merge-0) ; 
		\draw[merge node] (merge-2) to (merge-2l) ; 
		\draw[merge node] (merge-2) to (merge-3l) ; 
		\draw[merge node] (merge-1) to (merge-2) ; 
		\draw[merge node] (merge-3) to (merge-1) ; 
		\draw[merge node] (merge-4) to (merge-4l) ; 
		\draw[merge node] (merge-4) to (merge-5l) ; 
		\draw[merge node] (merge-3) to (merge-4) ; 
		\draw[mid interval] (1.000000,3.928571) -- ++(3.000000-1.000000,0) ;
		\draw[mid interval] (3.000000,5.857143) -- ++(4.000000-3.000000,0) ;
		\draw[mid interval] (4.000000,2.428571) -- ++(5.000000-4.000000,0) ;
		\draw[mid interval] (5.000000,8.214286) -- ++(7.500000-5.000000,0) ;
		\draw[mid interval] (7.500000,5.857143) -- ++(12.500000-7.500000,0) ;
		\tikzset{inactive/.style={draw opacity=.25}}

		\foreach \i/\x/\h/\k/\s in {
			1/{1/16}/2.428571/4/inactive,%
			2/{2/16}/3.928571/3/,%
			3/{3/16}/2.428571/4/inactive,%
			4/{4/16}/5.857143/2/,%
			5/{5/16}/2.428571/4/,%
			6/{6/16}/3.928571/3/inactive,%
			7/{7/16}/2.428571/4/inactive,%
			8/{8/16}/8.214286/1/,%
			9/{9/16}/2.428571/4/inactive,%
			10/{10/16}/3.928571/3/inactive,%
			11/{11/16}/2.428571/4/inactive,%
			12/{12/16}/5.857143/2/,%
			13/{13/16}/2.428571/4/inactive,%
			14/{14/16}/3.928571/3/inactive,%
			15/{15/16}/2.428571/4/inactive%
		} {
			\draw[comb,\s] ($(-.5,0)!\x!(15+.5,0)$) -- ++(0,\h) coordinate (ruler-\i) node[above] {\k} ;
		}
		\foreach \i in {1,...,16} {
			\draw[comb leaf,inactive] ($(-.5,0)!{(2*\i-1)/32}!(15+.5,0)$) -- ++(0,1.357143) coordinate (leaf-\i);
		}
		\foreach \i/\p in {
			1/2,%
			2/4,%
			3/2,%
			4/8,%
			5/6,%
			6/4,%
			7/6,%
			9/10,%
			10/12,%
			11/10,%
			12/8,%
			13/14,%
			14/12,%
			15/14%
		} {
			\draw[comb,very thin]  (ruler-\i) -- (ruler-\p) ;
		}
\draw[comb,very thin]  (ruler-1) -- (leaf-1) ;
\draw[comb,very thin]  (ruler-1) -- (leaf-2) ;
\draw[comb,very thin]  (ruler-3) -- (leaf-3) ;
\draw[comb,very thin]  (ruler-3) -- (leaf-4) ;
\draw[comb,very thin]  (ruler-5) -- (leaf-5) ;
\draw[comb,very thin]  (ruler-5) -- (leaf-6) ;
\draw[comb,very thin]  (ruler-7) -- (leaf-7) ;
\draw[comb,very thin]  (ruler-7) -- (leaf-8) ;
\draw[comb,very thin]  (ruler-9) -- (leaf-9) ;
\draw[comb,very thin]  (ruler-9) -- (leaf-10) ;
\draw[comb,very thin]  (ruler-11) -- (leaf-11) ;
\draw[comb,very thin]  (ruler-11) -- (leaf-12) ;
\draw[comb,very thin]  (ruler-13) -- (leaf-13) ;
\draw[comb,very thin]  (ruler-13) -- (leaf-14) ;
\draw[comb,very thin]  (ruler-15) -- (leaf-15) ;
\draw[comb,very thin]  (ruler-15) -- (leaf-16) ;
		\foreach \i in {0,...,4} {
			\draw[densely dotted,thick,black!50,->,shorten >=5pt] (B\i) ++(0,-.5) -- ++(0,-1.25) -| (merge-\i) ;
		}
	\end{scope}

	\begin{scope}[shift={(0,-8.214286-1)}]
		\draw[orange,thin] (0.000000-.5,-.2) -- ++(0,.4) ;
		\draw[orange,thin] (1.000000-.5,-.2) -- ++(0,.4) ;
		\draw[orange,thin] (2.000000-.5,-.2) -- ++(0,.4) ;
		\draw[orange,thin] (3.000000-.5,-.2) -- ++(0,.4) ;
		\draw[orange,thin] (4.000000-.5,-.2) -- ++(0,.4) ;
		\draw[orange,thin] (5.000000-.5,-.2) -- ++(0,.4) ;
		\draw[orange,thin] (6.000000-.5,-.2) -- ++(0,.4) ;
		\draw[orange,thin] (7.000000-.5,-.2) -- ++(0,.4) ;
		\draw[orange,thin] (8.000000-.5,-.2) -- ++(0,.4) ;
		\draw[orange,thin] (9.000000-.5,-.2) -- ++(0,.4) ;
		\draw[orange,thin] (10.000000-.5,-.2) -- ++(0,.4) ;
		\draw[orange,thin] (11.000000-.5,-.2) -- ++(0,.4) ;
		\draw[orange,thin] (12.000000-.5,-.2) -- ++(0,.4) ;
		\draw[orange,thin] (13.000000-.5,-.2) -- ++(0,.4) ;
		\draw[orange,thin] (14.000000-.5,-.2) -- ++(0,.4) ;
		\draw[orange,thin] (15.000000-.5,-.2) -- ++(0,.4) ;
		\draw[orange,thin] (16.000000-.5,-.2) -- ++(0,.4) ;
		\draw[orange] (-.5,0) -- ++(16,0) ;
	\end{scope}
\end{tikzpicture}%
	}
	\\[1ex]
	\adjustbox{max width=.9\linewidth}{%
		\begin{tikzpicture}[
		every node/.style = {font=\small},
		array entry/.style = {font=\footnotesize},
		power label/.style = {font=\footnotesize\itshape},
		run/.style = {draw=black,fill=black!5,rounded corners=1pt},
		merge result/.style = {run,draw=black!80!red,fill=black!80!red!10,opacity=.6},
		run leaf/.style = {run,draw=black!80!red,fill=black!80!red!30,opacity=.6},
		merge node/.style = {ultra thick,draw=black,fill=black!20,circle,opacity=.3,minimum size=14pt,},
		mid/.style = {densely dashed,blue!80!black!60},
		mid interval/.style = {{Arc Barb[arc=140,width=10pt]}-{Bracket[width=8pt]},shorten >=-.4pt,thick,blue!80!black},
		comb/.style = {line join=round,red!50!black,line width=4pt,opacity=.5},
		comb leaf/.style = {line join=round,red!50!black!40!gray,line width=4pt,opacity=.5},
		xscale=.6,yscale=.45,
]


	\begin{scope}[shift={(0,2)}]
		\draw[run] (0-.4,0) rectangle (2+.4,1) ;
		\draw[run] (3-.4,0) rectangle (3+.4,1) ;
		\draw[run] (4-.4,0) rectangle (4+.4,1) ;
		\draw[run] (5-.4,0) rectangle (5+.4,1) ;
		\draw[run] (6-.4,0) rectangle (9+.4,1) ;
		\draw[run] (10-.4,0) rectangle (15+.4,1) ;
		\draw[mid] (1.000000,1.1) -- ++(0,-10.100000-3.1) ;
		\draw[mid] (3.000000,1.1) -- ++(0,-10.100000-3.1) ;
		\draw[mid] (4.000000,1.1) -- ++(0,-10.100000-3.1) ;
		\draw[mid] (5.000000,1.1) -- ++(0,-10.100000-3.1) ;
		\draw[mid] (7.500000,1.1) -- ++(0,-10.100000-3.1) ;
		\draw[mid] (12.500000,1.1) -- ++(0,-10.100000-3.1) ;
		\node[power label] (B0) at (2+0.5,1.25) {2} ;
		\node[power label] (B1) at (3+0.5,1.25) {1} ;
		\node[power label] (B2) at (4+0.5,1.25) {2} ;
		\node[power label] (B3) at (5+0.5,1.25) {1} ;
		\node[power label] (B4) at (9+0.5,1.25) {1} ;
		\node[array entry] at (0,0.5) {14} ;
		\node[array entry] at (1,0.5) {15} ;
		\node[array entry] at (2,0.5) {16} ;
		\node[array entry] at (3,0.5) {13} ;
		\node[array entry] at (4,0.5) {12} ;
		\node[array entry] at (5,0.5) {11} ;
		\node[array entry] at (6,0.5) {7} ;
		\node[array entry] at (7,0.5) {8} ;
		\node[array entry] at (8,0.5) {9} ;
		\node[array entry] at (9,0.5) {10} ;
		\node[array entry] at (10,0.5) {1} ;
		\node[array entry] at (11,0.5) {2} ;
		\node[array entry] at (12,0.5) {3} ;
		\node[array entry] at (13,0.5) {4} ;
		\node[array entry] at (14,0.5) {5} ;
		\node[array entry] at (15,0.5) {6} ;
	\end{scope}

	\begin{scope}[shift={(0,-10.100000)}]
		\draw[merge result] (0-0.3,10.100000-0.4) rectangle ++(16-0.4,0.8) ;
		\draw[merge result] (0-0.3,5.900000-0.4) rectangle ++(4-0.4,0.8) ;
		\draw[run leaf] (0-0.3,0.500000-0.4) rectangle ++(3-0.4,0.8) ;
		\draw[run leaf] (3-0.3,0.500000-0.4) rectangle ++(1-0.4,0.8) ;
		\draw[merge result] (4-0.3,5.900000-0.4) rectangle ++(2-0.4,0.8) ;
		\draw[run leaf] (4-0.3,0.500000-0.4) rectangle ++(1-0.4,0.8) ;
		\draw[run leaf] (5-0.3,0.500000-0.4) rectangle ++(1-0.4,0.8) ;
		\draw[run leaf] (6-0.3,0.500000-0.4) rectangle ++(4-0.4,0.8) ;
		\draw[run leaf] (10-0.3,0.500000-0.4) rectangle ++(6-0.4,0.8) ;
		\node[merge node] (merge-3) at (7.500000,10.100000) {} ;
		\node[merge node] (merge-1) at (3.500000,10.100000) {} ;
		\node[merge node] (merge-0) at (2.500000,5.900000) {} ;
		\node[] (merge-0l) at (1.000000,0.500000) {} ;
		\node[] (merge-1l) at (3.000000,0.500000) {} ;
		\node[merge node] (merge-2) at (4.500000,5.900000) {} ;
		\node[] (merge-2l) at (4.000000,0.500000) {} ;
		\node[] (merge-3l) at (5.000000,0.500000) {} ;
		\node[merge node] (merge-4) at (11.500000,10.100000) {} ;
		\node[] (merge-4l) at (7.500000,0.500000) {} ;
		\node[] (merge-5l) at (12.500000,0.500000) {} ;
		\draw[merge node] (merge-0) to (merge-0l) ; 
		\draw[merge node] (merge-0) to (merge-1l) ; 
		\draw[merge node] (merge-1) to (merge-0) ; 
		\draw[merge node] (merge-2) to (merge-2l) ; 
		\draw[merge node] (merge-2) to (merge-3l) ; 
		\draw[merge node] (merge-1) to (merge-2) ; 
		\draw[merge node] (merge-3) to (merge-1) ; 
		\draw[merge node] (merge-4) to (merge-4l) ; 
		\draw[merge node] (merge-4) to (merge-5l) ; 
		\draw[merge node] (merge-3) to (merge-4) ; 
		\draw[mid interval] (1.000000,5.900000) -- ++(3.000000-1.000000,0) ;
		\draw[mid interval] (3.000000,10.100000) -- ++(4.000000-3.000000,0) ;
		\draw[mid interval] (4.000000,5.900000) -- ++(5.000000-4.000000,0) ;
		\draw[mid interval] (5.000000,10.100000) -- ++(7.500000-5.000000,0) ;
		\draw[mid interval] (7.500000,10.100000) -- ++(12.500000-7.500000,0) ;
		\tikzset{inactive/.style={draw opacity=.25}}

		\foreach \i/\x/\h/\k/\s in {
			1/{1/16}/5.900000/2/inactive,%
			2/{2/16}/5.900000/2/inactive,%
			3/{3/16}/5.900000/2/,%
			4/{4/16}/10.100000/1/,%
			5/{5/16}/5.900000/2/,%
			6/{6/16}/5.900000/2/inactive,%
			7/{7/16}/5.900000/2/inactive,%
			8/{8/16}/10.100000/1/,%
			9/{9/16}/5.900000/2/inactive,%
			10/{10/16}/5.900000/2/inactive,%
			11/{11/16}/5.900000/2/inactive,%
			12/{12/16}/10.100000/1/,%
			13/{13/16}/5.900000/2/inactive,%
			14/{14/16}/5.900000/2/inactive,%
			15/{15/16}/5.900000/2/inactive%
		} {
			\draw[comb,\s] ($(-.5,0)!\x!(15+.5,0)$) -- ++(0,\h) coordinate (ruler-\i) node[above] {\k} ;
		}
		\foreach \i in {1,...,16} {
			\draw[comb leaf,inactive] ($(-.5,0)!{(2*\i-1)/32}!(15+.5,0)$) -- ++(0,2.900000) coordinate (leaf-\i);
		}
		\foreach \i/\p in {
			1/2,%
			2/4,%
			3/2,%
			4/8,%
			5/6,%
			6/4,%
			7/6,%
			9/10,%
			10/12,%
			11/10,%
			12/8,%
			13/14,%
			14/12,%
			15/14%
		} {
			\draw[comb,very thin]  (ruler-\i) -- (ruler-\p) ;
		}
\draw[comb,very thin]  (ruler-1) -- (leaf-1) ;
\draw[comb,very thin]  (ruler-1) -- (leaf-2) ;
\draw[comb,very thin]  (ruler-3) -- (leaf-3) ;
\draw[comb,very thin]  (ruler-3) -- (leaf-4) ;
\draw[comb,very thin]  (ruler-5) -- (leaf-5) ;
\draw[comb,very thin]  (ruler-5) -- (leaf-6) ;
\draw[comb,very thin]  (ruler-7) -- (leaf-7) ;
\draw[comb,very thin]  (ruler-7) -- (leaf-8) ;
\draw[comb,very thin]  (ruler-9) -- (leaf-9) ;
\draw[comb,very thin]  (ruler-9) -- (leaf-10) ;
\draw[comb,very thin]  (ruler-11) -- (leaf-11) ;
\draw[comb,very thin]  (ruler-11) -- (leaf-12) ;
\draw[comb,very thin]  (ruler-13) -- (leaf-13) ;
\draw[comb,very thin]  (ruler-13) -- (leaf-14) ;
\draw[comb,very thin]  (ruler-15) -- (leaf-15) ;
\draw[comb,very thin]  (ruler-15) -- (leaf-16) ;
		\foreach \i in {0,...,4} {
			\draw[densely dotted,thick,black!50,->,shorten >=5pt] (B\i) ++(0,-.5) -- ++(0,-1.25) -| (merge-\i) ;
		}
	\end{scope}

	\begin{scope}[shift={(0,-10.100000-1)}]
		\draw[orange,thin] (0.000000-.5,-.2) -- ++(0,.4) ;
		\draw[orange,thin] (1.000000-.5,-.2) -- ++(0,.4) ;
		\draw[orange,thin] (2.000000-.5,-.2) -- ++(0,.4) ;
		\draw[orange,thin] (3.000000-.5,-.2) -- ++(0,.4) ;
		\draw[orange,thin] (4.000000-.5,-.2) -- ++(0,.4) ;
		\draw[orange,thin] (5.000000-.5,-.2) -- ++(0,.4) ;
		\draw[orange,thin] (6.000000-.5,-.2) -- ++(0,.4) ;
		\draw[orange,thin] (7.000000-.5,-.2) -- ++(0,.4) ;
		\draw[orange,thin] (8.000000-.5,-.2) -- ++(0,.4) ;
		\draw[orange,thin] (9.000000-.5,-.2) -- ++(0,.4) ;
		\draw[orange,thin] (10.000000-.5,-.2) -- ++(0,.4) ;
		\draw[orange,thin] (11.000000-.5,-.2) -- ++(0,.4) ;
		\draw[orange,thin] (12.000000-.5,-.2) -- ++(0,.4) ;
		\draw[orange,thin] (13.000000-.5,-.2) -- ++(0,.4) ;
		\draw[orange,thin] (14.000000-.5,-.2) -- ++(0,.4) ;
		\draw[orange,thin] (15.000000-.5,-.2) -- ++(0,.4) ;
		\draw[orange,thin] (16.000000-.5,-.2) -- ++(0,.4) ;
		\draw[orange] (-.5,0) -- ++(16,0) ;
	\end{scope}
\end{tikzpicture}%
	}

	\caption{%
		Example merge tree for Powersort (top) and 4-way Powersort (bottom) for an input of size $n=16$: 
		Run boundaries are mapped (gray arrows) to nodes of a (virtual) perfectly balanced binary resp.\ 4-ary tree $T_V$ (shown in light red),
		which has the array elements as its leaves.
		The midpoints of two adjacent runs form the horizontal range in which this node is allowed to move; it then ``snaps'' to the highest pole in that range.
	}
	\label{fig:powersort-example-n16}
\end{figure}
\begin{figure*}[tbhp]
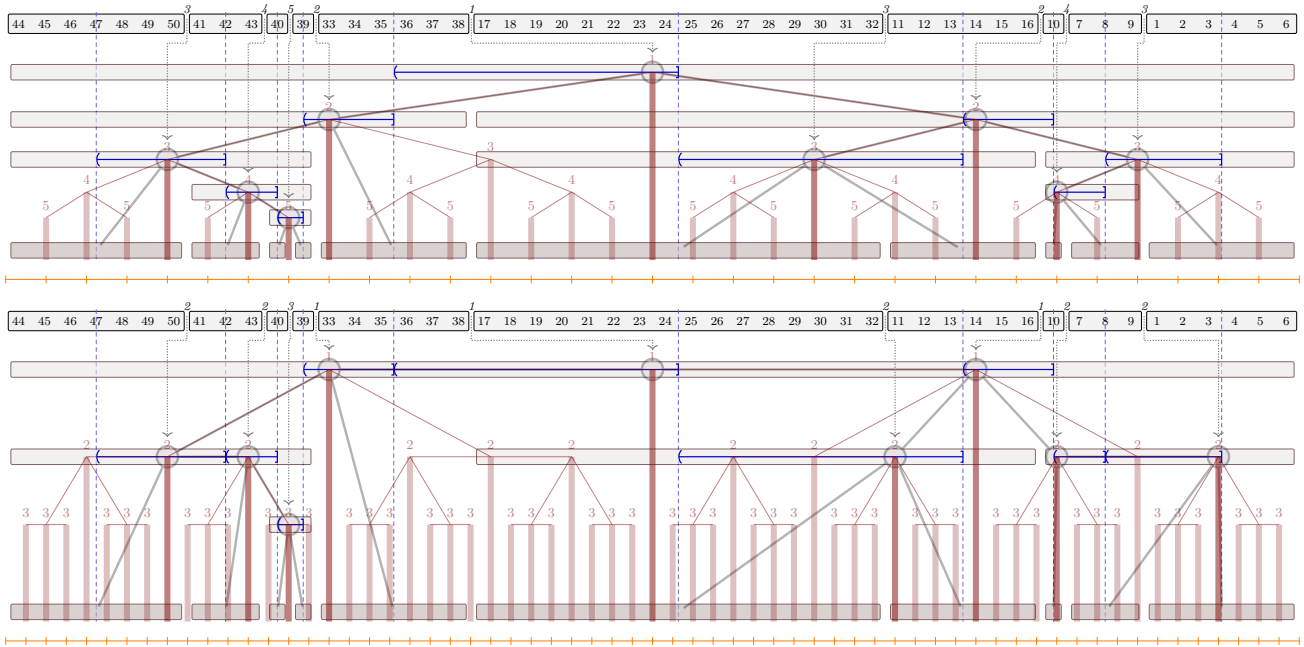

	\adjustbox{max width=\textwidth}{%
		\input{powersort-n50-2way.tikz}%
	}
	\\[1ex]
	\adjustbox{max width=\textwidth}{%
		\input{powersort-n50-4way.tikz}%
	}

	\caption{%
		Example merge tree for Powersort (top) and 4-way Powersort (bottom) for an input of size $n=50$: 
		Here, merges are ``rounded'' to nodes in a (virtual) perfectly balanced binary resp.\ 4-ary tree $T_V$ (shown in light red);
		drawing style as in \wref{fig:powersort-example-n16}.
	}
	\label{fig:powersort-example}
\end{figure*}

We now introduce Powersort's merge policy for general $k$-way merging.
We expressly include the special case $k=2$ here; the resulting merge policy is equivalent 
to Powersort from~\cite{MunroWild2018}.

\subsection{Intuition and Merge Tree}

Let $k\ge 2$ be fixed, and first assume for ease of presentation that $n = k^m$, $m\in\N$.
The intuition behind Powersort is to imagine a complete $k$-ary tree~$T_V$ sitting on top of the array,
where each element $A[i]$ of the input corresponds to a leaf of~$T_V$.
$T_V$ corresponds to the merge tree of a non-adaptive $k$-way Mergesort (starting with individual elements).
Now each run boundary \runboundary j also corresponds to an (inner) node $w$ of this (virtual) $k$-ary tree~$T_V$,
namely the lowest common ancestor of the two leafs adjacent to \runboundary j.

We cannot follow $T_V$ if we want to retain all existing runs (most will usually straddle node boundaries); 
but we will use the depths $d_w$ of the nodes of $T_V$ as a guide:
Any non-leaf node $w$ in $T_V$ splits (a range of) the array into $k$ segments.  
At the $k-1$ boundaries between these segments, we draw a vertical bar of height $m-d_w$ (see \wref{fig:powersort-example-n16}).
Then, we assign each boundary \runboundary j its \emph{midpoint interval}: the horizontal interval
from the middle of \leafnode{j-1} (exclusive) up to the middle of \leafnode j (inclusive).
Midpoint intervals are shown in blue in \wref{fig:powersort-example-n16}.
Now, the \emph{power} $P_j$ of \runboundary j is defined as the depth $d_w$ of the node $w$ contributing the tallest vertical bar inside $B_j$'s midpoint interval (\wref{def:node-power}).

Since the midpoint intervals are disjoint, each vertical bar is assigned at most one \runboundary j,
and so each node $w$ in $T_V$ is assigned at most $k-1$ run boundaries.
These $k-1$ run boundaries define one merge-tree node~$v(w)$; 
the collection of $v(w)$ forms the $k$-way merge tree $T$ of Powersort.
Intuitively, $T$ is a good merge tree since it uses boundaries as close to the perfectly balanced $T_V$
while respecting existing runs.
Note that going from $2$-way powers to $4$-way powers as in \wref{fig:powersort-example-n16}
corresponds to a simple transformation: 
We ``squish'' pairs of adjacent levels of $2$-way powers into one layer of $4$-way powers: 
$P^{(4)}_i = \lfloor (P^{(2)}_i-1)/2 \rfloor +1$.
We thus obtain the merge tree of $4$-way Powersort from the one for $2$-way Powersort by letting all nodes
at even depth absorb their children (and adopt their grandchildren).

When $n$ is not a power of $k$, the virtual tree $T_V$ does not align with element boundaries, 
but we can treat the array and its runs as \emph{continuous} intervals; 
the assignment and argument from above work the same way; see \wref{fig:powersort-example}
for an example.
Formally, we map the array to the unit interval, partitioned by the runs to define the power of run boundaries:

\begin{definition}[Run Boundary Power]
\label{def:node-power}
	Let $k\ge 2$ be fixed.
	For run lengths $L_0,\ldots,L_{r-1}$, 
	set $\ell_i = L_i / n$.
	For $1\le i < r$,
	let \runboundary i be the run boundary between the $(i-1)$st and $i$th run.
	The \thmemph{($k$-way) power} of \runboundary i is
	\begin{align*}%
			P^{(k)}_i
		&\rel=
			\min \Bigl\{ 
				p\in\N : \bigl\lfloor a_i \cdot k^{p} \bigr\rfloor 
					< \bigl\lfloor b_i \cdot k^{p} \bigr\rfloor  
			\Bigr\},\;\;
	\\[-.5ex] &
	\text{where }
			a_i
		\rel=
			\sum_{j=0}^{i-1} \ell_j - \tfrac12 \ell_{i-1},\;\;
			b_i
		\rel=
			\sum_{j=0}^{i-1} \ell_j + \tfrac12 \ell_{i}.
	\end{align*}
\end{definition}
We drop the superscript when $k$ is clear from context.
\wref{fig:power} illustrates \wref{def:node-power}.

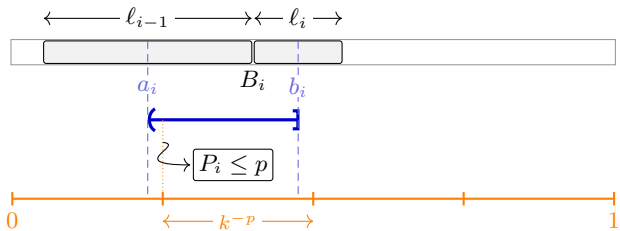
\begin{figure}[bth]
	\begin{tikzpicture}[
				every node/.style = {font=\small},
				run/.style = {draw=black,fill=black!5,rounded corners=1pt,inner sep=0pt},
				mid/.style = {densely dashed,blue!80!black!60},
				mid interval/.style = {{Arc Barb[arc=140,width=10pt]}-{Bracket[width=8pt]},shorten >=-.5pt,very thick,blue!80!black},
				comb/.style = {line join=round,blue,line width=4pt,opacity=.3},
				xscale=8,yscale=.3,
		]
			\def\yskip{5}
			\def\l{.05}
			\def\m{.4}
			\def\r{.55}
			\def\a{.225}
			\def\b{.475}

			\node[draw=black!40,inner sep=.5pt,fit={(0,0) (1,1)}] {};
			
			\foreach \f/\t in 
				{\l/\m,\m/\r} 
			{
				\node[run,fit={(\f+.002,0) (\t-.002,1)}] {} ;
			}
			\draw[dotted] (\m,1) ++(0,-1.75) node[] {\runboundary i} ;
			
			\begin{scope}[shift={(0,-6)},thick,orange]
				\draw[-] (0,0) -- (1,0);
				\foreach \x/\l in {0/0,1/1} {
					\draw (\x,.3) -- ++ (0,-.6) ++(0,-1) node[anchor=base] {\l} ;
				}
				\foreach \i in {1,...,3} {
					\draw ({\i/4},.3) -- ++ (0,-.6) ++(0,-1) 
						;
				}
				\draw[thin,<->,shorten >=0pt, shorten <=0pt] 
					(.5,-1) -- node[fill=white,inner sep=2pt] {\smaller $k^{-p}$} (.25,-1) ;
			\end{scope}
			
			\begin{scope}[shift={(0,-2.5)}]
				\draw[mid interval] (\a,0) -- (\b,0) ;
			\end{scope}
			
			\foreach \x/\l in {\a/a,\b/b} {
				\draw[mid] ({\x},1) -- ({\x},-6);
				\node[mid,fill=white,inner sep=1pt] at (\x,-1) {$\like{\l}{\l_{i}}$} ;
			}
			
		\foreach \f/\t/\L in {\l/\m/$\ell_{i-1}$,\m/\r/$\ell_i$} {
			\draw[<->,shorten >=.5pt, shorten <=.5pt] 
				(\f,2) -- node[fill=white,inner sep=2pt] {\L} (\t,2) ;
		}
		
		\draw[orange,densely dotted] (.25,-6) -- (.25,-2.5) ;
		\draw[->,shorten >=1pt] (.25,-3.5) to[out=-70,in=180,looseness=.3] (.3,-4.5) 
				node[draw,right,inner sep=2pt,rounded corners=1pt] {$P_i \le p$} ;
	
	\end{tikzpicture}
	\caption{%
		Illustration of the power of a run boundary.
		Since there is an orange tick inside the midpoint interval, we obtain $P^{(k)}_i \le p$;
		in the example, the midpoint interval is $(a_i,b_i] = (0.225,0.475]$, $k=4$ and $p=1$, so $P^{(4)}_i = 1$.
	}
	\label{fig:power}
\end{figure}

Based on (just) $P_1,\ldots,P_{r-1}$, we can recover the merge tree $T$
by recursively finding all occurrences of the smallest occurring power
and using these run boundaries as the splitting points for recursion.
The pseudocode in \wref{alg:merge-tree} makes this concrete.

\begin{algorithm}[tbhp]
\hspace*{2em}\begin{minipage}{0.9\linewidth}%
\small
\tikzset{every node/.style={font=\scriptsize}}
	\begin{codebox}
		\Procname{$\proc{KwayTree}_k(A[b..e))$}
	\li		Let $s_0 = b,s_1,\ldots,s_{r-1},s_r = e$ be 
	\zi		\qquad the start indices of runs in $A[b..e)$
	\li		Let $P_1,\ldots,P_{r-1}$ be the boundary powers
	\li		$\{i_1,\ldots,i_m\}$ = $\arg \min \{ P_i : i\in[1..r)\}$
	\li		$i_0$ = $s_0$; $i_{m+1}$ = $s_r$
	\li		\For $j=0,\ldots,m$ 
	\Do
		\li		$\proc{KwayTree}_k(A[i_{j}..i_{j+1}))$ \Comment recurse
	\EndFor
	\li		$\proc{Merge}(A[i_0..i_1),\ldots,A[i_m..i_{m+1}))$
	\end{codebox}	
\end{minipage}
\caption{Conceptual algorithm to compute the merge tree of $k$-way Powersort.}
\label{alg:merge-tree}
\end{algorithm}

This is just a conceptual algorithm;
Multiway Powersort will produce the same merge tree as $\proc{KwayTree}_k(A[0..n))$, 
but without having to search for minimal powers explicitly or even knowing all runs and powers 
at any point in time (as described below).
$\proc{KwayTree}_k$ makes it obvious that $P_j$ is the ``intended'' depth of 
$m(\runboundary j)$
since we combine nodes recursively by increasing power. 
Nodes can end up \emph{higher} in the merge tree if not all power values occur.
Since we count edges for the depth, but powers start at 1, this yields following useful inequality.

\begin{fact}\label{fact:boundary-power}
	$\mathrm{depth}(m(\runboundary j)) \le P_j - 1$.
\end{fact}

A simple proof by induction shows that in each recursive call of $\proc{KwayTree}_k(A[0..n))$,
we have $m \le k-1$, so we are merging at most $k$ runs at once.

\subsection{k-way Powersort}

Like Timsort, Powersort processes the merge tree left to right.
We maintain a stack~$S$ of runs and powers, subject to the invariant
that powers are weakly increasing in $S$ (from bottom to top).
When the next run boundary \runboundary i has a power $P_i$ larger or equal than $S.\id{top}()$, 
we push $(\leafnode{i-1},P_i)$ onto $S$ and continue.
Otherwise, we merge \leafnode i with \emph{all runs} on $S$ with power \emph{equal} to $S.\id{top}()$;
this is a key difference to standard Powersort, where we always just merge with $S.\id{top}()$.
We repeatedly perform such merges until we can push $(\leafnode{i-1},P_i)$ onto $S$ 
without violating the invariant.

The definition of $k$-way powers ensures that there are at most $k-1$ equal powers on $S$ at any time,
so all merges executed in the above process combine at most $k$ runs.
Once all run boundaries have been treated that way, the remaining stack is merged top to bottom (similar to Timsort)

\begin{algorithm}[tbph]
\hspace*{1em}\begin{minipage}{0.9\linewidth}%
\small
\tikzset{every node/.style={font=\scriptsize}}
	\begin{codebox}
		\Procname{$\proc{KwayPowerSort}_k(A[0..n))$}
		\li $S \gets $ empty stack  \quad \Comment\textsmaller{capacity $(k-1)\lceil\log_k(n)+1\rceil$}
		\li $b_1 \gets 0$; \; 
			$e_1 = \proc{FirstRunOf}(A[b_1..n))$
		\li \While $e_1 < n$
		\Do
			\li $b_2 \gets e_1$; \;
				$e_2 \gets \proc{FirstRunOf}(A[b_2..n))$
			\li $P \gets \proc{Power}_k(n,b_1,e_1,b_2,e_2)$ 
			\li \While $S.top().power > P$ %
			\Do
				\li $P' \gets S.top().power$
				\li $L \gets$ empty list;\;   $L.\id{append}(S.\id{pop}())$ %
				\li \While $S.\id{top}().power \isequal P'$
				\Do
					\li $L.\id{append}(S.\id{pop}())$
				\EndWhile
					\zi \Comment merge runs in $L$ with $A[b_1..e_1)$%
				\li $(b_1,e_1) \gets \proc{Merge}(L, \, A[b_1..e_1))$
			\EndWhile
			\li $S.\id{push}(A[b_1,e_1), P)$
			\li $b_1 \gets b_2$; \; $e_1 \gets e_2$
		\EndWhile
		\Comment Now $A[b_1..e_1)$ is the rightmost run
		\li \While $\neg S.\id{empty}()$
		\Do
			\zi \Comment pop (up to) $k-1$ runs, merge with $A[b_1..e_1)$
			\li $(b_1,e_1) \gets \proc{Merge}(S.\id{pop}(k-1), \, A[b_1..e_1))$
		\EndWhile
	\end{codebox}
	\vspace{2ex}
	\begin{codebox}
		\Procname{$\proc{Power}_k(n,b_1,e_1,b_2,e_2)$}
		\li $n_1 \gets e_1 - b_1$; \; 
			$n_2 \gets e_2 - b_2$; \; $p \gets 0$
		\li $a \gets ( b_1 + n_1/2 ) / n$; \;
			$b \gets ( b_2 + n_2/2 ) / n$
		\li \kw{while} $\lfloor a\cdot k^p \rfloor \isequal \lfloor b\cdot k^p \rfloor$ 
			\kw{do} $p\gets p+1$ \kw{end while}
		\li \Return $p$
	\end{codebox}
	\end{minipage}
	\caption{%
		Multiway Powersort pseudocode.
		The function \proc{FirstRunOf} finds the leftmost run in an array and returns its endpoint.
	}
	\label{alg:powersort}
	\vspace*{-3ex}
\end{algorithm}%

\subsection{Analysis}

We state a few properties about Multiway Powersort, most importantly a bound on its merge cost.
The case $k=2$ has been covered in~\cite{MunroWild2018}, but the following self-contained proof
is arguably more direct.

\begin{theorem}
\label{thm:kway-powersort-analysis}
	The merge cost of \proc{$k$-way PowerSort} is  
	$M \le \frac1{\lg(k)}\mathcal H(\frac{L_0}n,\ldots,\frac{L_{r-1}}n)n+2n$ and
	the number of comparisons $C \le \frac{\lceil \lg k\rceil}{\lg k} \mathcal H(\frac{L_0}n,\ldots,\frac{L_{r-1}}n)n+(1+2\frac{\lceil \lg k\rceil}{\lg k})n+(k-1)r$.
\end{theorem}

\begin{proof}
We actually show $M/n \le \sum_{i=0}^{r-1} \ell_i \lceil \log_k(1/\ell_i) + 1 \rceil$, which implies the claim.
We start with the expression from \wref{sec:run-adaptive-mergesort} for the merge cost: 
$M=\sum_{i=0}^{r-1} d_i \cdot L_i$ where $d_i$ is the depth of leaf \leafnode i in~$T$.
Since \leafnode i is a leaf, it must be the (direct) child of either $m(\runboundary i)$ or $m(\runboundary{i+1})$ in $T$ and hence
$d_i = 1 + \max\{\mathrm{depth}(m(\runboundary i)), \mathrm{depth}(m(\runboundary{i+1}))\}$;
by formally setting $\mathrm{depth}(m(\runboundary 0)) = \mathrm{depth}(m(\runboundary r)) = 0$, 
this equality holds for all $i=0,\ldots,r-1$.
Now, \wref{fact:boundary-power} yields $d_i \le \hat P_i$, $i=0,\ldots,r-1$, where 
$\hat P_i \colonequals \max\{P_i, P_{i+1}\}$ and we similarly set $P_0 = P_r = 0$.
So we have 
$M/n \le \sum_{i=0}^{r-1} \ell_i \hat P_i$.

We will now show that $k^{-p} \le \ell_i / 2$ implies $P^{(k)}_i \le p$.
It is easy to see that with $a_i$ and $b_i$ as in \wref{def:node-power}, we have $\smash{P^{(k)}_i} = \min\{p : \exists c\in\N:c\cdot k^{-p}\in(a_i,b_i]\}$; (cf.\ \wref{fig:power}).
It suffices to note that $|(a_i,b_i]| \ge \ell_i/2$ and whenever $k^{-p} \le |(a_i,b_i]|$ is given,
$(a_i,b_i]$ clearly contains a value $c k^{-p}$ (an orange ``tick'' in \wref{fig:power}).
Note that since $|(a_{i+1},b_{i+1}]| \ge \ell_{i}/2$ similarly holds, the above argument also shows that 
$k^{-p} \le \ell_{i} / 2$ implies $P^{(k)}_{i+1} \le p$.

Now, $k^{-p} \le \ell_i / 2$ iff $p \ge \log_k(1/\ell_i)+1$, so setting $p = \lceil\log_k(1/\ell_i) + 1\rceil$
implies that condition and hence $P^{(k)}_i \le \lceil\log_k(1/\ell_i)\rceil + 1$ as well as $P^{(k)}_{i+1} \le \lceil\log_k(1/\ell_i)\rceil + 1$;
hence $\hat P_i \le \lceil\log_k(1/\ell_i)\rceil + 1$.
Inserting shows
$M/n \le \sum_{i=0}^{r-1} \ell_i \lceil\log_k(1/\ell_i)+1\rceil$ as claimed.

For the number of comparisons, recall that the tournament-tree merge uses $\le \lceil\lg(k)\rceil$ comparisons per output element, except for the first output element of each merge, where we initialize the tournament using $k-1$ comparisons. 
Let $\mu$ be the number of merges; the exact number depends on the merge policy, but $\lceil \frac{r-1}{k-1}\rceil \le \mu \le r-1$.
Including the $n-1$ comparisons to find runs, 
we overall spend 
$
		C
	\le 
		M\cdot \lceil \lg k\rceil + \mu\cdot (k-1-\lceil \lg k\rceil) + n-1
$.
Inserting the bound for $M$, we obtain the claimed
$
		C
	\le 
		\frac{\lceil \lg k\rceil}{\lg k} \cdot (\mathcal H n + 2n) + (k-1)r + n
$ 
comparisons.
\end{proof}

\begin{remark}
\begin{enumerate}[label=(\alph*)]
\item 
	We note that the analysis of $C$ could be slightly sharpened, but we are mostly interested in the case where $k$ is a small power of two; there, the above bound is tight up to lower order terms.
\item 
	For $k=2$, the bounds simplify to Thm.~5 of~\cite{MunroWild2018},
	for which we above give an alternative proof that 
	is entirely self-contained and elementary.
\end{enumerate}
\end{remark}

\begin{proposition}
	The maximal stack height for $k$-way Powersort
	is $(k-1)\lceil\log_k(n)+1\rceil$.
\end{proposition}

\begin{proof}
	We use that $P_i \le \lceil\log_k(1/\ell_i)+1\rceil$ (see above);
	since $\ell_i \ge 1/n$, we obtain $1\le P_i\le \lceil\log_k(n)+1\rceil$.
	So there are $\lceil\log_k(n)+1\rceil$ different possible values for~$P_i$.
	The run stack $S$ is weakly increasing and can contain at most $k-1$ entries with equal power,
	so $|S|\le \lceil\log_k(n)+1\rceil(k-1)$.
\end{proof}

\begin{remark}
	We point out that Peeksort, the other algorithm from~\cite{MunroWild2018},
	can also be generalized to $k$-way merges.
    In essence, instead  of finding the (single) run boundary closest to the array midpoint, we now use $k-1$ equidistant points and pick the run boundaries closest to those as boundaries of recursive calls. There are a few corner cases that require care, though.
	We defer our corresponding results to a full version of this paper due to space constraints and since
    Powersort is the preferred method in practice.

\end{remark}

\section{Results}

This section describes our empirical results around Multiway Powersort.

\subsection{Implementations and setup}

We have implemented $2$-way and $4$-way Powersort in \Cpp and engineered both 
for performance under the GNU Compiler Collection. 
We also compare them to a few standard algorithms including \texttt{std::sort} and \texttt{std::stable\_sort} from the GNU implementation of the \Cpp Standard Library.
Key parts of the code are shown in \wref{app:code}.
The full implementations including instructions to reproduce our running-time study are available on GitHub: \url{https://github.com/sebawild/powersort}.

We comment on a few implementation details that are relevant for the running-time comparisons.

\paragraph{Tournament tree}
For 4-way merging, we use a tournament tree with 3 internal nodes.
We represent it as a winner tree, \ie, each node stores the run contributing the smaller element 
among its two children;
more specifically, each node store a pointer to the first element of a run that has not been output yet.
The root additionally stores whether the minimum there came from the left or right subtree.
That way each element moved to the output entails 2 tournaments (node recomputations), and we know which ones these are.
Since our tournament tree is small and will entirely reside in registers, directly accessing the two children of a node for playing a match is very inexpensive, 
so there seems to be little benefit from using a loser tree instead of the simpler winner tree.
Whenever an element of a run enters the tournament tree (because it has won against its sibling run),
we already advance the pointer for the leaves of the tournament tree; this decouples the advancement of the pointer from writing the output.

\paragraph{Sentinel values}
We implemented merging methods that assume the availability of a $+\infty$ value that is strictly larger
than any element in the list to be sorted.
Such values are often available; floating-point types directly support them, integer types may be restricted to not using the largest representable value; strings and characters can use the null character 
(with a symmetric implementation that uses values smaller than any value to sort).
We can place such a sentinel value at the end of runs to avoid a pointer comparison in the inner loop of the merging methods, which saves a few instructions and branches per iteration.
Below, we report on results for algorithms with and without sentinels.

\paragraph{Sentinel-free merging by stages}
When a type does not (efficiently) support a $+\infty$ value, 
it can be beneficial to have a merge method that does not use sentinels.
For 2-way Powersort, a trivial such merge implementation is only slightly slower overall than using sentinel-based 2-way merge.
For 4-way merging using tournament trees, there are many more cases to distinguish (which runs are empty),
making a trivial implementation inefficient.
We combine two tricks to improve upon that: First, we merge in ``stages'' where a stage ends whenever one run is exhausted. 
This allows us to reduce the number of cases to distinguish in the code (at the expense of longer code).
Second, if $s$ is the length of the shortest remaining (nonempty) run, we can clearly do at least $s$ iterations without checking boundaries, after which we recompute $s$.
When $s$ reaches $0$, we know that the current stage is over.
``4-way Powersort (no sentinel)'' uses merging by stages.

\paragraph{Buffer for runs}
All our merge methods have an in-place interface, \ie, they produce the merged result in the array initially occupied by the input runs (and assume those to be adjacent), but they use a linear-size buffer for efficient
merging. 
Most methods copy all runs to the buffer (which tends to give the fastest code for the merging phase),
but we also include a 2-way method that only copies the smaller run to the buffer and merges ``into the gap'' (a standard trick used, \eg, in Timsort). This method does not use sentinels.

\paragraph{Minimal run length}
All our Mergesort variants use Insertionsort on subproblems of size $\le 24$;
for Timsort-based methods like Powersort, this means that when a new run is found with fewer elements,
it is extended to 24 elements.
\subsection{Experimental Setup}
The experiments were run on a PC running Ubuntu 20.04.4 LTS. The CPU has 4 cores running 2 hyperthreads each and 16\,GB of main memory. The model of the CPU is Intel(R) Core(TM) i7-4790 CPU @ 3.60GHz. 
Each core has a private level-1 cache with 32KB for data and 32KB for instructions and a unified level-2 cache
of 256KB.
Moreover, there's a shared level-3 cache of 8MB.
All caches are organized into 64B cache lines and are 8-way associative; the level-3 cache is 16-way associative.

The experiments use \texttt{g++} from the GNU Compiler Collection 9.4.0 (Ubuntu 9.4.0-1ubuntu1~20.04.1)
with optimization flags \texttt{-Ofast -march=native -mtune=native}.
(Results did not change noticeably when using \texttt{-O3} instead.)

\subsection{Hypothesis 1: 4-way Powersort can yield significant performance improvements}

To test this hypothesis, 
we conducted a running-time study varying several dimensions to investigate the speedup obtainable from
4-way merging in Powersort.

As an indicative first scenario, \wref{fig:norm-times-int-runs} shows normalized running times for sorting a 
mildly presorted list of \texttt{int}s (as in~\cite{MunroWild2018}): across four orders of magnitude of input sizes, 4-way Powersort
is consistently around 20\% faster.
\wref{fig:dist-times-int-runs} shows that the distribution of running times is concentrated around the mean,
showing a robust difference.
\wref{fig:norm-times-int-runs} also shows that adaptive mergesort methods outperform the Quicksort-based \texttt{std::sort}.

The random-runs inputs have substantial sorted parts; in expectation $\sqrt n$ runs with an expected $\sqrt n$ elements each, albeit with substantial variation.
However, the lower bound $n \mathcal H$ approaches $\frac12 n \lg(n)$, so they are still far from fully sorted;
in particular, sorting these still has linearithmic complexity (as clearly visible in \wref{fig:norm-times-int-runs}).

\begin{figure}[btp]
	\includegraphics[width=\linewidth]{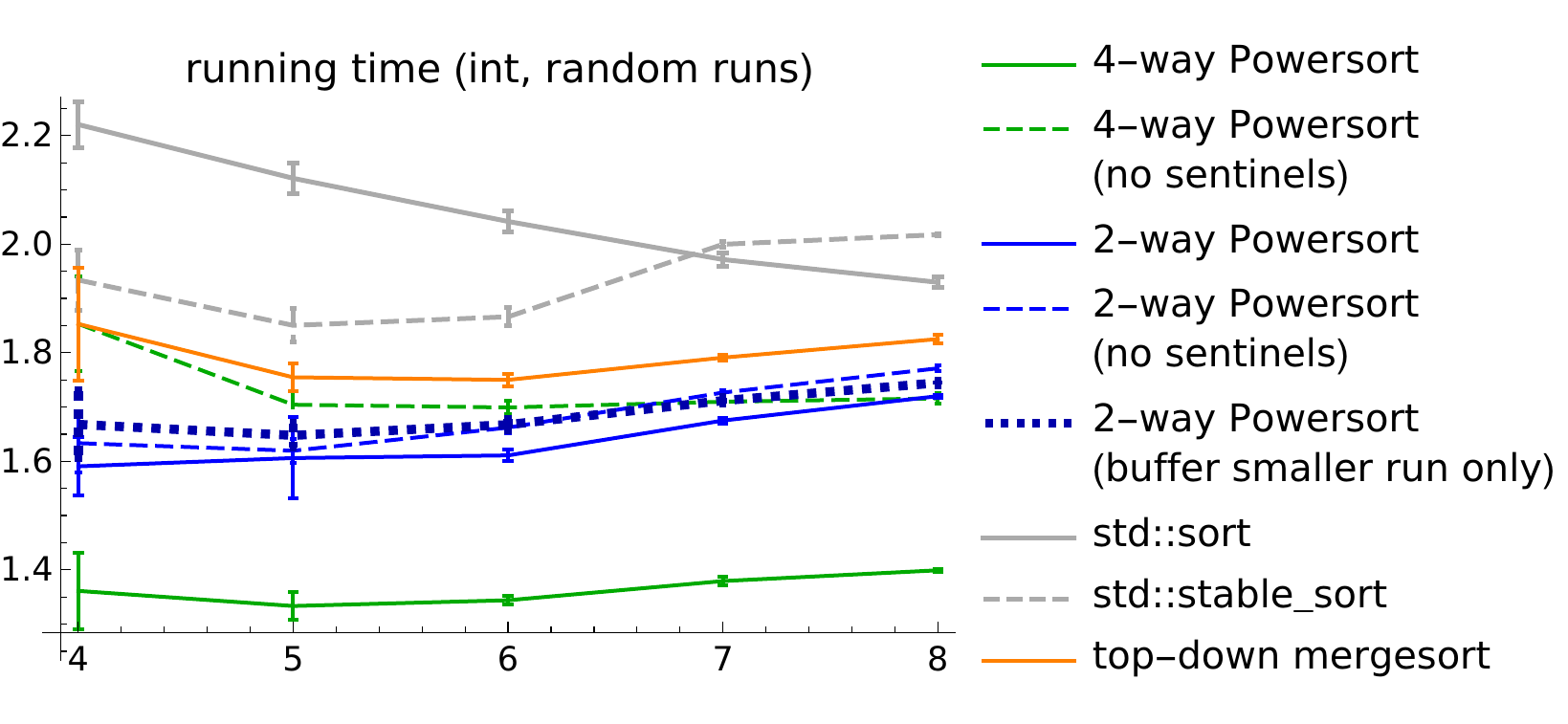}
	\caption{%
		Normalized \textbf{running times} for sorting \textbf{\texttt{int}s} (4 byte signed integers)
		where the input has \textbf{random} initial \textbf{runs} with a geometric distribution and expected length $\sqrt n$.
		The $x$-axis shows $\log_{10}(n)$, the $y$-axis shows average running time in ms multiplied by ${10^6}/{n \lg n}$; error bars show one standard deviation. We used 1000 repetitions up to $10^6$ elements and 100 repetitions for $10^7$ and $10^8$ elements.
	}
	\label{fig:norm-times-int-runs}
\end{figure}

\begin{figure}[btp]
	\includegraphics[width=\linewidth]{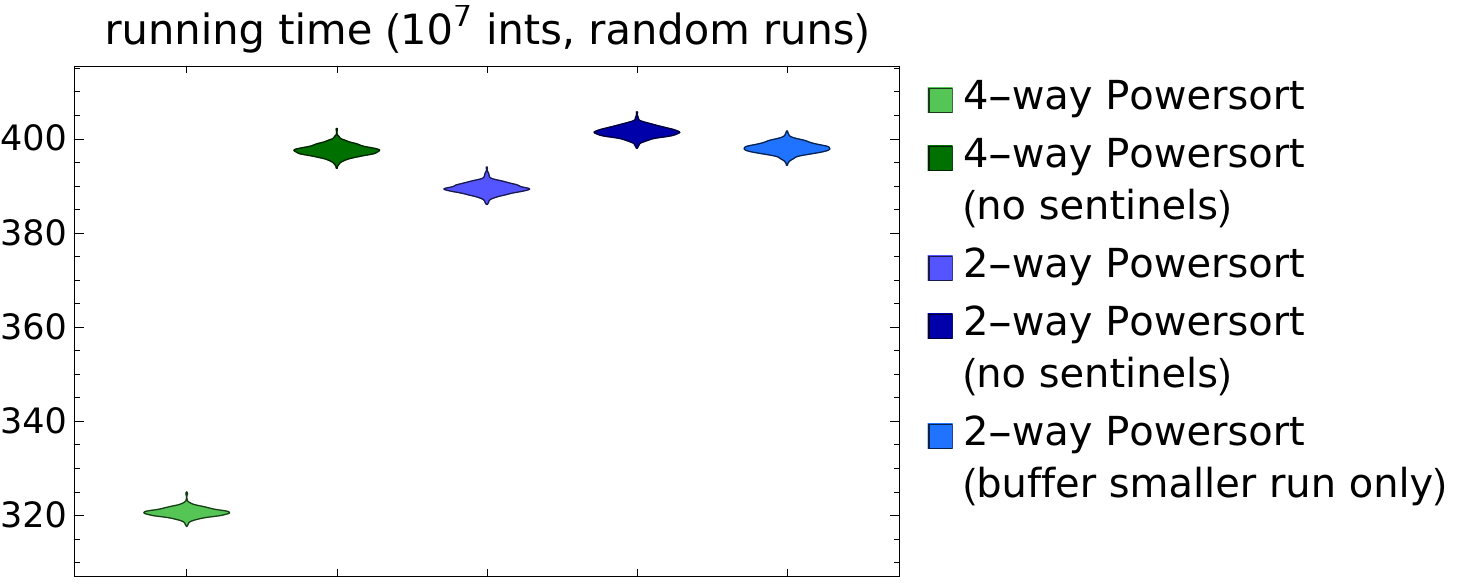}
	\caption{%
		Distribution of running times for $n=10^7$ from \wref{fig:norm-times-int-runs} (sorting \texttt{int}s with random runs).
	}
	\label{fig:dist-times-int-runs}
\end{figure}

\begin{figure}[tbp]
	\includegraphics[width=\linewidth]{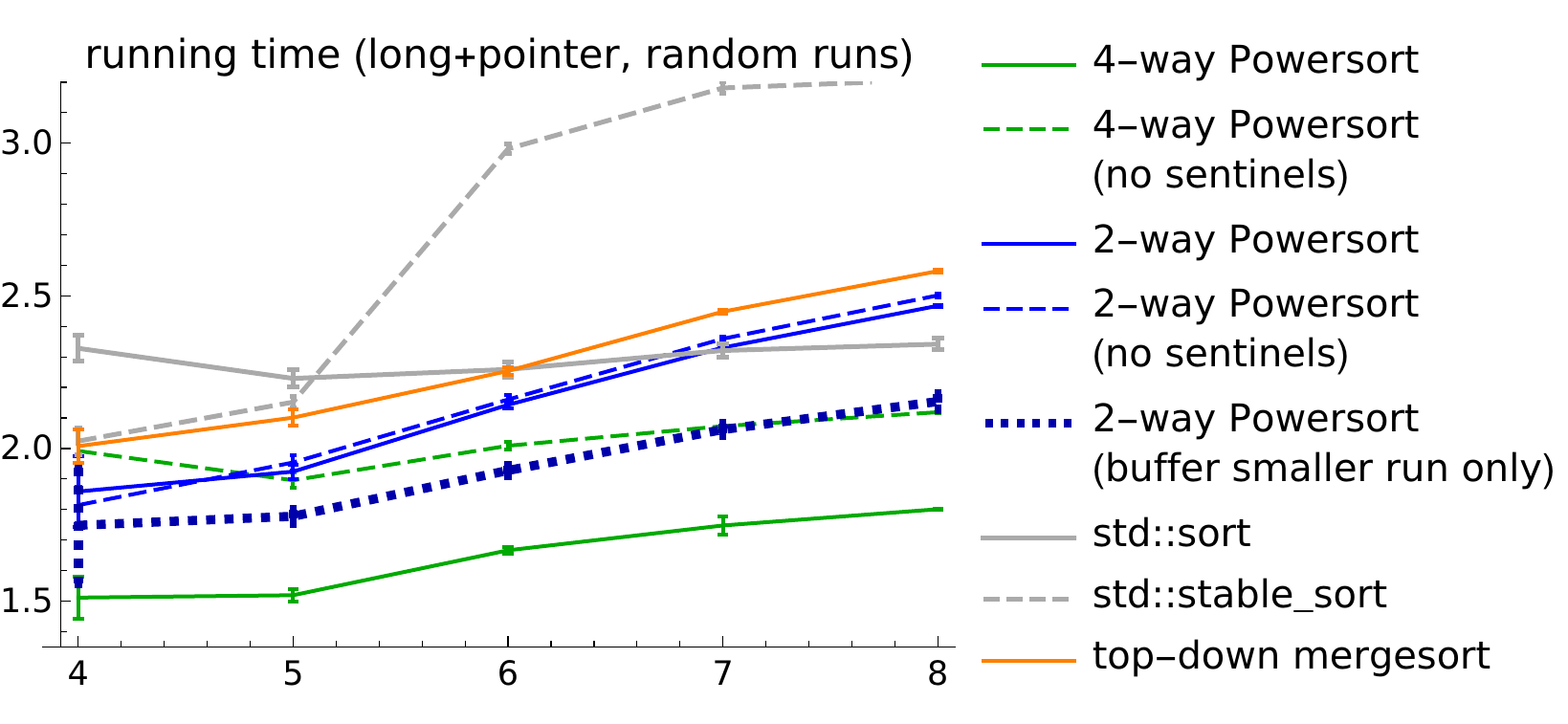}
	\caption{%
		Normalized \textbf{running times} for sorting \textbf{records} of a \texttt{long} key and a pointer (16 byte in total)
		where the input has \textbf{random} initial \textbf{runs} with a geometric distribution and expected length $\sqrt n$.
		Axes as in \wref{fig:norm-times-int-runs}.
	}
	\label{fig:norm-times-l+p-runs}
\end{figure}

A first variation considers a more realistic data type than integers:
Stable sorting is most relevant for sorting objects that have various properties / instance variables 
and which are (partially) sorted according to one of its properties.
We model this by sorting records containing a \texttt{long} key and a pointer.
\wref{fig:norm-times-l+p-runs} shows the respective running times.

Since elements are now 4 times larger (16 vs.\ 4 bytes), 
it is not surprising that all methods are slower than for \texttt{int}s, but they are not uniformly so.
4-way Powersort with the sentinel-free merging by stages was roughly as fast as 2-way Powersort, which in turn
had essentially equal performance under all three considered 2-way merge methods.
For records, the 4-way method and the 2-way method that copies only the smaller run to a buffer, substantially outperform the other 2-way Powersort variants. 
4-way Powersort with sentinel-based merge remains substantially faster still (between 15 and 20\% across the studied input-size range).

\begin{figure}[tbp]
	\includegraphics[width=\linewidth]{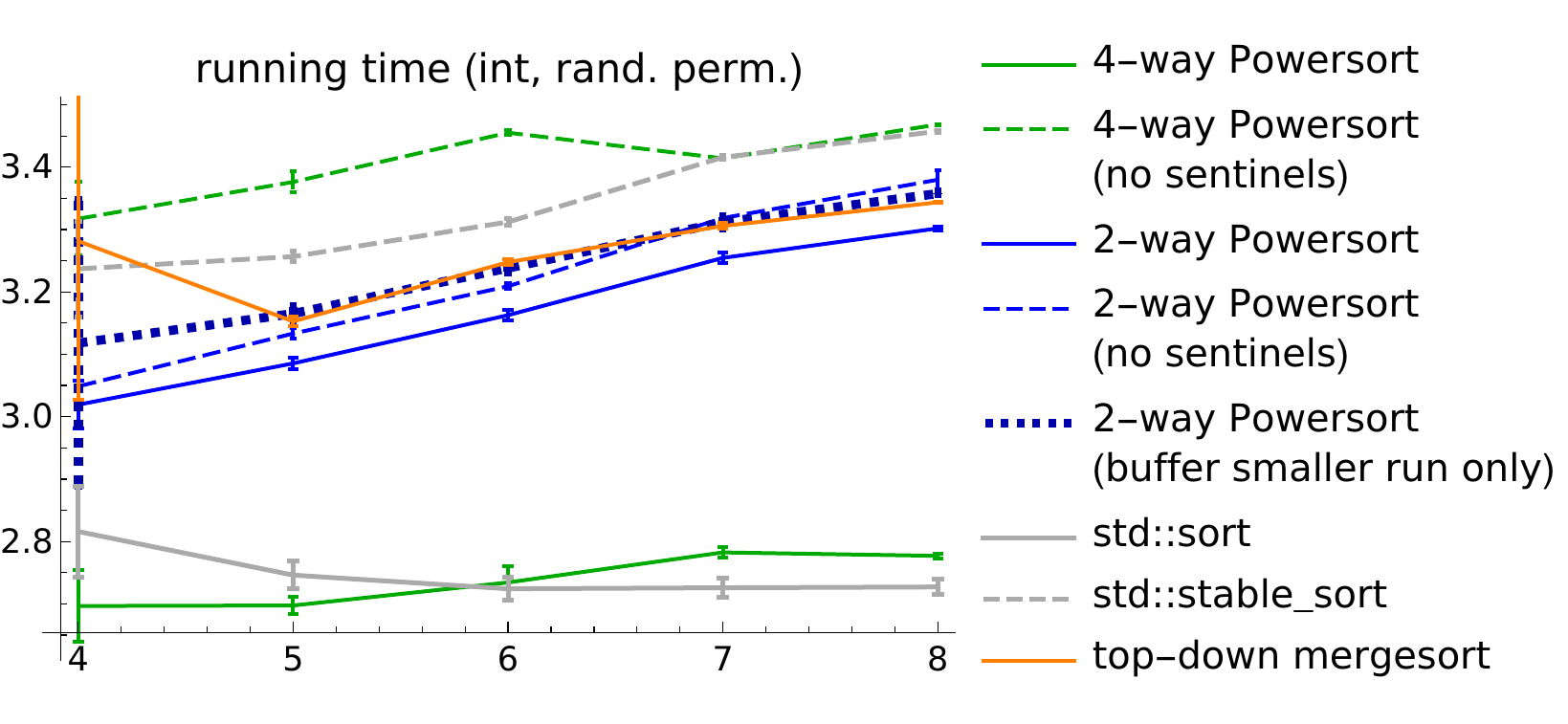}
	\caption{%
		Normalized \textbf{running times} for sorting \texttt{int}s
		where the input is a \textbf{random permutation} of $n$.
		The $x$-axis shows $\log_{10}(n)$, the $y$-axis shows average running time in ms multiplied by ${10^6}/{n \lg n}$; error bars show one standard deviation.
	}
	\label{fig:norm-times-int-rp}
\end{figure}

Finally, we vary the input distribution.
In particular, any adaptive sorting method should remain competitive with the best general purpose methods when no significant presortedness exists to be exploited.
\wref{fig:norm-times-int-rp} shows running times for random permutations.
Somewhat surprisingly, even where effectively no presortedness exists, 4-way Powersort is still 15--20\% faster than the fastest 2-way method and almost as fast as \texttt{std::sort}.

We have taken care to implement the merge methods (2-way and multiway) as efficiently as possible.
The presented methods are chosen from a much larger collection of implementation alternatives we have experimented with.
As \wref[Figures]{fig:norm-times-int-runs}--\ref{fig:norm-times-int-rp} show,
a $+\infty$ value, that can be used as a sentinel value for stopping loops 
(where otherwise a range check for the iterator would be needed),
can yield a substantial speedup in multiway merging.

\subsection{Hypothesis 2: Scanned elements explain the speedups}

\begin{figure}[tbp]
	\includegraphics[width=\linewidth]{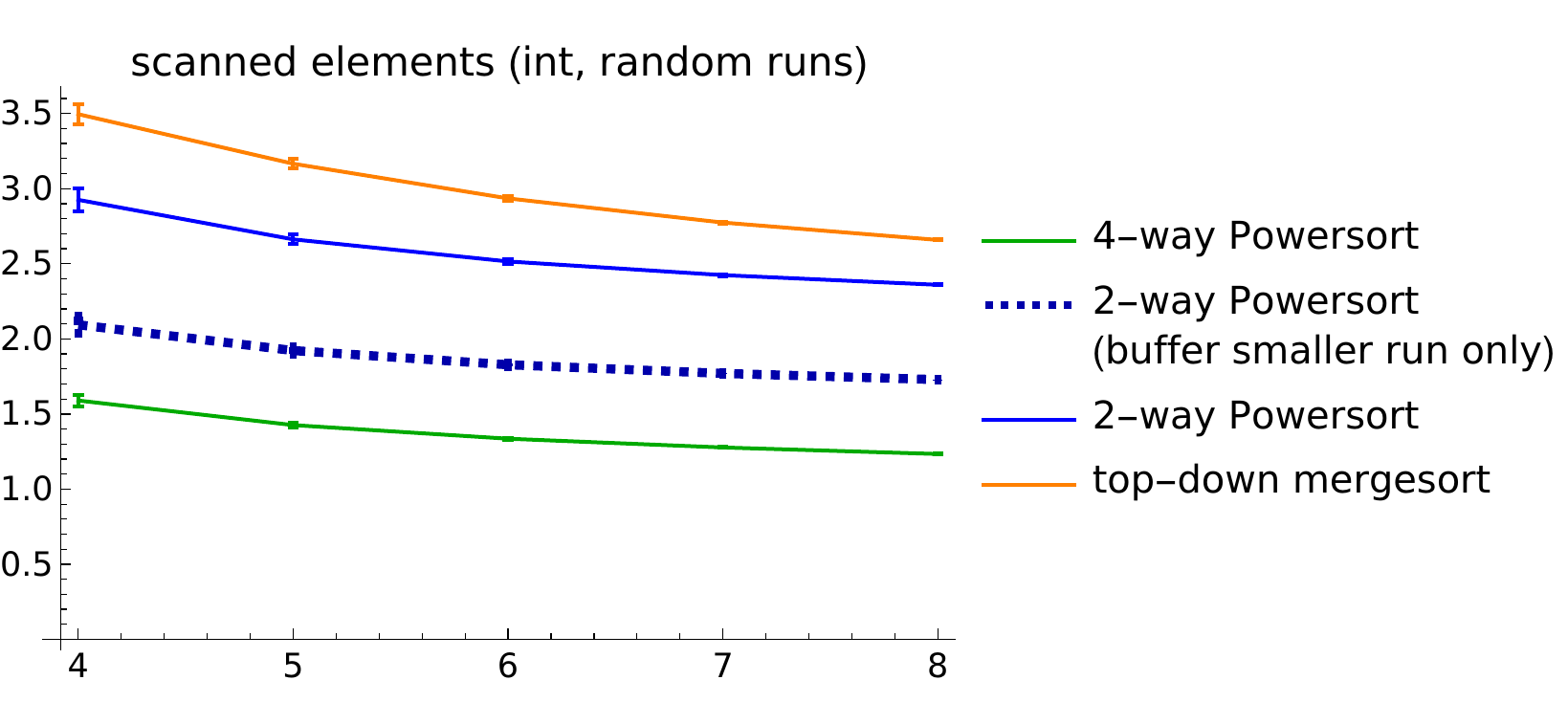}
	\caption{%
		Normalized \textbf{scanned elements} when the input has random initial runs with a geometric distribution and expected length $\sqrt n$.
		The $x$-axis shows $\log_{10}(n)$, the $y$-axis shows average merge cost divided by $n \lg(n/24)$. 
		Error bars show one standard deviation.
	}
	\label{fig:norm-scanned-elements-runs}
\end{figure}

\begin{table*}[tbp]
	\adjustbox{max width=\linewidth}{%
	\begin{tabular}{lrrrrrrr}
	\toprule
		\textbf{Algorithm} & \textbf{L1 rm} & \textbf{L1 wm} & \textbf{instructions} & \textbf{cycle estimate} & \textbf{merge cost} & \textbf{buffer cost} & \textbf{comparisons} \\
	\midrule
		nop                &             17 &      6\,250\,001 &           400\,000\,361 &           1\,087\,503\,111 &                 --- &                  --- &                  --- \\
	\midrule
		4P                 &     90\,998\,034 &     90\,941\,221 &        14\,386\,631\,421 &          26\,834\,656\,351 &         678\,233\,797 &          678\,244\,262 &        1\,401\,885\,151 \\
		4P w/o sent            &     91\,104\,862 &     90\,991\,531 &        21\,788\,074\,703 &          34\,238\,304\,273 &         678\,233\,797 &          678\,233\,797 &        1\,401\,804\,460 \\
		2P                 &    168\,472\,841 &    168\,392\,912 &        15\,593\,213\,610 &          36\,636\,357\,390 &       1\,298\,329\,585 &        1\,298\,349\,759 &        1\,398\,341\,256 \\
		2P w/o sent            &    124\,242\,042 &    118\,893\,488 &        17\,678\,030\,619 &          32\,327\,153\,489 &       1\,298\,329\,585 &          603\,947\,459 &        1\,398\,280\,296 \\
	\bottomrule
	\end{tabular}%
	}
	\caption{
		Cachegrind results for sorting one input of $n=10^8$ \texttt{int}s with random runs of expected length $\sqrt n = 10^4$.
		The table shows level 1 cache read misses (L1 rm), level 1 cache write misses (L1 wm), 
		the total number of executed instructions, and cachegrind's total-time estimate, defined as $\text{cycles} = \text{instr} + 10 \cdot \text{L1-misses} + 100\cdot \text{LL-misses}$, where LL = level 3 cache.
		The algorithms are
		``nop'' = no sorting at all, 
		``4P'' = 4-way Powersort, 
		``2P'' = 2-way Powersort,
		``4P w/o sent'' = 4-way Powersort without sentinels,
		``2P w/o sent'' = 2-way Powersort without sentinels and copying only the smaller run to the buffer.
		Here, ``nop'' is used to get a baseline; the reads come from checking that the input is correctly sorted as part of our running time setup.
		The reported numbers do not include the cost of generating the random input.
		The L1 cache lines can hold 16 elements (for our machine); read and write misses are determined by cachegrind.
		The last three columns were determined directly from the code using counters.
	}
	\label{tab:cachegrind}
\end{table*}

\begin{figure}
	\includegraphics[width=\linewidth]{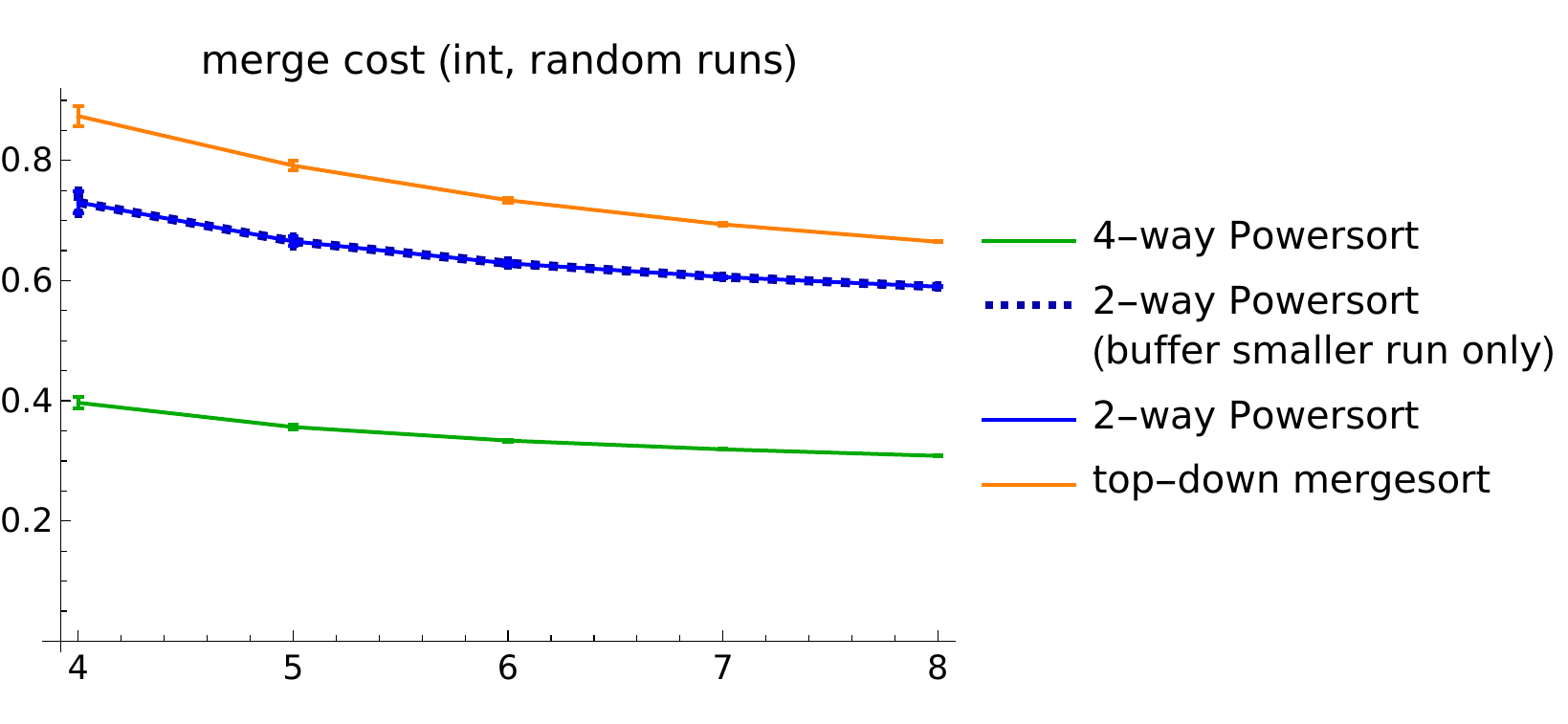}
	\caption{%
		Normalized \textbf{merge cost} when the input has random initial runs with a geometric distribution and expected length $\sqrt n$.
		The $x$-axis shows $\log_{10}(n)$, the $y$-axis shows average merge cost divided by $n \lg(n/24)$. 
		Error bars show one standard deviation.
	}
	\label{fig:norm-merge-cost-int-runs}
\end{figure}

After documenting the speedup realized by 4-way Powersort, we now explore likely explanations for these observations.
To that end, we ran our code through the cachegrind/callgrind memory hierarchy simulator for valgrind;
\wref{tab:cachegrind} shows the results.
While there is an 8\% saving in the plain number of executed instructions~-- 
owing to the extremely lean inner loops in our sentinel-based 4-way method~-- 
the much more dramatic improvement lies in the number of cache misses:
these are reduced to 54\% of the 2-way method that copies both runs to buffer resp.\ 73\% when copying the smaller only.
Our data shows that 4-way Powersort simultaneously reduces both the number of instruction and the number of cache misses.

We point out that all memory accesses in our sorting methods happen in sequential scans, 
so the memory access patterns are already as cache- and prefetcher friendly as can be;
instead of cache misses, rather the sheer data volume transferred from memory is the issue.
Since the memory scans in merging are either read-only or write-only, the total number of scanned elements
can be computed very accurately:
On the used machine, a cache line has 64 bytes, \ie, 16 \texttt{int}s, so we obtain scanned element counts as $16\cdot(\text{L1rm}+ \text{L1wm})$ (\wref{fig:norm-scanned-elements-runs}).
As a sanity check, we can compare numbers as follows: Each merge with merge cost $m$ entails $2m$ read and $2m$ writes, moreover, we do (in total) one scan over the input for run detection ($n$ reads) and initialize the buffer ($n$ writes); this yields 100.034\% of the reported L1 read misses and 100.097\% of L1 write misses for 4P (slightly more since a few reads are cached).

We use cachegrind's simple linear-combination model as a first-order approximation for explaining running times from scanned elements and instruction counts; see column cycle estimate in~\wref{tab:cachegrind}.
Comparing \wref{fig:norm-times-int-runs} (at $n=10^8$) to the cycle estimate, we see that despite its simplicity, this model is able to correctly identify 4P as the most efficient variant and it even gives
a usable approximation of the relative improvement over the 2-way methods.
While it does not rank the remaining algorithms (whose running time is very similar) correctly,
we conclude that the large reduction in scanned elements is the most likely reason for the speedup in 4-way Powersort. 
However, reducing scanned elements at the expense of a substantial increase in executed instructions is not fruitful.

\subsection{Hypothesis 3: 4-way Powersort halves merge cost}

Last, we tie the reduction in scanned elements to the key algorithmic innovation.
\wref{fig:norm-merge-cost-int-runs} shows the merge cost over different input sizes.

We note that for worst-case inputs and as $\mathcal H \to \infty$, 
\wref{thm:kway-powersort-analysis} shows 
that the merge cost of 4-way Powersort is asymptotically half of that for 2-way Powersort.
Whether that is possible to realize for practical values of $n$ is not clear a priori, though.
As \wref{fig:norm-merge-cost-int-runs} shows, 4-way Powersort reduces the merge cost to around 52\% of 
that of Powersort, very close to halving costs.

\section{Conclusions}
Powersort, introduced in 2018 by Munro \& Wild, is a stable (binary) Mergesort variant that exploits existing runs and finds nearly-optimal merging orders with negligible overhead. 
In this paper we presented Multiway Powersort, a generalization of Powersort, that uses $k$-way merges in order to speedup the algorithm. 
While the extension of classic 2-way Mergesort to $k$-way merges is straightforward, 
doing the same with optimal run-adaptive performance is an algorithmic challenge.
In solving that, we also generalize the underlying algorithm for nearly-optimal \emph{binary} search trees to nearly-optimal $k$-way search trees, which might be of independent interest.

The motivation for our generalization are expected advantages of multiway merges for the memory hierarchy of modern computers. We conducted extensive experiments based on our highly engineered implementation of 4-way Powersort around following hypotheses: Compared to 2-way Powersort
\begin{enumerate}
	\item 4-way Powersort can yield significant performance improvements;
	\item Scanned elements are a good predictor for this speedup;
	\item 4-way Powersort halves the merge cost.
\end{enumerate}
Our experiments fully confirmed the first hypothesis. 
Considering presorted integer inputs, 4-way Powersort was about 20\% faster than the 2-way variant. 
This speedup was robust in terms of the size of the data items to be sorted (record data instead of integers) as well as the amount of presortedness. 
However, significant speedups are only realized when the data type allows sentinels to be used. If these are not available (because there is no largest possible value), the running time advantages are made up for by additional range checks. 

Regarding the second hypothesis, we find that compared to different versions of 2-way Powersort, 4-way Powersort incurs between 54\% and 73\% of the cache misses, and scanned-element counts very accurately reflect that. 
This is a plausible explanation for the observed speedup.
A simple linear combination of the number of executed instructions and scanned elements explains much (but not all) of performance differences.

Last but not least, we find that the reduced flexibility of doing 4-way merges instead of 2-way merges hardly affects the effectiveness of Powersort's nearly-optimal merge policy: The merge cost is reduced to 52\% for moderate-sized inputs, close to the theoretical optimum of $\frac12$ approached in the limit.

In conclusion, 4-way Powersort provides notable advantages over both 2-way Powersort, and other known stable sorting algorithms. Thus, Multiway Powersort is a strong contender for the fastest general, stable sorting algorithm -- especially in cases where presorted data can be expected. 
Both Mergesort variants use $\Theta(n)$ words of extra space for buffer areas; 
while binary merges can easily work with a buffer for $n/2$ elements, the same is not easily achieved with 4-way merging, where our current implementation uses space $n$ elements.

Future work will study whether similar speedups are observed when sorting pointers to objects (as always happens in CPython). Moreover, when comparisons are expensive, combining multiway merging with the galloping merge strategy could be useful.
Another interesting question is whether values of $k$ larger than 4 offer some benefit.
On one hand, a further reduction of memory transfers is likely beneficial for large inputs;
on the other hand, the merging method becomes more complicated and thus implies more overheads.
It is also conceivable that with larger $k$, the fraction of merge cost that arises
from non-full merges, \ie, merging $<k$ runs, would grow.

	\myacknowledgements
%

%
\bibliography{references}

\begin{thebibliography}{10}

\bibitem{AugerJugeNicaudPivoteau2018}
Nicolas Auger, Vincent Jug{\'e}, Cyril Nicaud, and Carine Pivoteau.
\newblock On the worst-case complexity of {TimSort}.
\newblock In Hannah~Bast Yossi~Azar and Grzegorz Herman, editors, {\em 26th
  Annual European Symposium on Algorithms (ESA 2018)}, Leibniz International
  Proceedings in Informatics (LIPIcs), 2018.
\newblock \href {https://doi.org/10.4230/LIPIcs.ESA.2018.4}
  {\path{doi:10.4230/LIPIcs.ESA.2018.4}}.

\bibitem{AugerJugeNicaudPivoteau2018arxiv}
Nicolas Auger, Vincent Jugé, Cyril Nicaud, and Carine Pivoteau.
\newblock On the worst-case complexity of timsort, 2018.
\newblock URL: \url{https://arxiv.org/abs/1805.08612}, \href
  {https://doi.org/10.48550/ARXIV.1805.08612}
  {\path{doi:10.48550/ARXIV.1805.08612}}.

\bibitem{AumullerDietzfelbinger2016}
Martin Aumüller, Martin Dietzfelbinger, and Pascal Klaue.
\newblock How good is multi-pivot quicksort?
\newblock {\em ACM Transactions on Algorithms}, 13(1):1--47, dec 2016.
\newblock \href {https://doi.org/10.1145/2963102} {\path{doi:10.1145/2963102}}.

\bibitem{BarbayNavarro2013}
J{\'{e}}r{\'{e}}my Barbay and Gonzalo Navarro.
\newblock On compressing permutations and adaptive sorting.
\newblock {\em Theoretical Computer Science}, 513:109--123, November 2013.
\newblock \href {https://doi.org/10.1016/j.tcs.2013.10.019}
  {\path{doi:10.1016/j.tcs.2013.10.019}}.

\bibitem{BussKnop2019}
Sam Buss and Alexander Knop.
\newblock Strategies for stable merge sorting.
\newblock In {\em Symposium on Discrete Algorithms (SODA 2019)}, pages
  1272--1290. SIAM, jan 2019.
\newblock \href {https://doi.org/10.1137/1.9781611975482.78}
  {\path{doi:10.1137/1.9781611975482.78}}.

\bibitem{DeGouwBoerBubelHaehnleRotSteinhoefel2017}
Stijn de~Gouw, Frank~S. de~Boer, Richard Bubel, Reiner H{\"{a}}hnle, Jurriaan
  Rot, and Dominic Steinh{\"{o}}fel.
\newblock Verifying {OpenJDK's} sort method for generic collections.
\newblock {\em Journal of Automated Reasoning}, August 2017.
\newblock \href {https://doi.org/10.1007/s10817-017-9426-4}
  {\path{doi:10.1007/s10817-017-9426-4}}.

\bibitem{GhasemiJugeKhalighinejad2022}
Elahe Ghasemi, Vincent Jug\'{e}, and Ghazal Khalighinejad.
\newblock Galloping in fast-growth natural merge sorts.
\newblock In Miko{\l}aj Boja\'{n}czyk, Emanuela Merelli, and David~P. Woodruff,
  editors, {\em International Colloquium on Automata, Languages, and
  Programming (ICALP 2022)}, volume 229 of {\em Leibniz International
  Proceedings in Informatics (LIPIcs)}, pages 68:1--68:19, Dagstuhl, Germany,
  2022. Schloss Dagstuhl -- Leibniz-Zentrum f{\"u}r Informatik.
\newblock \href {https://doi.org/10.4230/LIPIcs.ICALP.2022.68}
  {\path{doi:10.4230/LIPIcs.ICALP.2022.68}}.

\bibitem{Juge2020}
Vincent Jug{\'{e}}.
\newblock Adaptive shivers sort: An alternative sorting algorithm.
\newblock In {\em Symposium on Discrete Algorithms (SODA 2020)}, pages
  1639--1654. SIAM, jan 2020.
\newblock URL: \url{https://doi.org/10.1137%2F1.9781611975994.101}, \href
  {https://doi.org/10.1137/1.9781611975994.101}
  {\path{doi:10.1137/1.9781611975994.101}}.

\bibitem{Knuth1998}
Donald~E. Knuth.
\newblock {\em {The Art Of Computer Programming: Searching and Sorting}}.
\newblock Addison Wesley, 2nd edition, 1998.

\bibitem{KushagraLopezOrtizQiaoMunro2013}
Shrinu Kushagra, Alejandro L{\'{o}}pez-Ortiz, Aurick Qiao, and J.~Ian Munro.
\newblock Multi-pivot quicksort: Theory and experiments.
\newblock In {\em Workshop on Algorithm Engineering and Experiments
  ({ALENEX})}, pages 47--60. SIAM, dec 2013.
\newblock \href {https://doi.org/10.1137/1.9781611973198.6}
  {\path{doi:10.1137/1.9781611973198.6}}.

\bibitem{MunroWild2018}
J.~Ian Munro and Sebastian Wild.
\newblock Nearly-optimal mergesorts: Fast, practical sorting methods that
  optimally adapt to existing runs.
\newblock In Yossi Azar, Hannah Bast, and Grzegorz Herman, editors, {\em
  European Symposium on Algorithms (ESA)}, volume 112 of {\em LIPIcs}, pages
  63:1--63:16. Schloss Dagstuhl--Leibniz-Zentrum fuer Informatik, 2018.
\newblock URL: \url{https://www.wild-inter.net/publications/munro-wild-2018},
  \href {https://doi.org/10.4230/LIPIcs.ESA.2018.63}
  {\path{doi:10.4230/LIPIcs.ESA.2018.63}}.

\bibitem{NebelWildMartinez2016}
Markus~E. Nebel, Sebastian Wild, and Conrado Mart{\'i}nez.
\newblock Analysis of pivot sampling in dual-pivot {Quicksort}~-- a holistic
  analysis of yaroslavskiy’s partitioning scheme.
\newblock {\em Algorithmica}, 75(4):632--683, August 2016.
\newblock URL:
  \url{https://www.wild-inter.net/publications/nebel-wild-martinez-2016}, \href
  {http://arxiv.org/abs/1412.0193} {\path{arXiv:1412.0193}}, \href
  {https://doi.org/10.1007/s00453-015-0041-7}
  {\path{doi:10.1007/s00453-015-0041-7}}.

\bibitem{Peters2002mailinglist}
Tim Peters.
\newblock {[Python-Dev] Sorting}, 2002.
\newblock URL:
  \url{https://mail.python.org/pipermail/python-dev/2002-July/026837.html}.

\bibitem{Peters2021}
Tim Peters.
\newblock {Timsort (listsort.txt)}, 2021.
\newblock URL:
  \url{https://github.com/python/cpython/blob/main/Objects/listsort.txt#L343}.

\bibitem{PetersEtAl2018}
Tim Peters et~al.
\newblock Replace list sorting \texttt{merge\_collapse()}?, 2018.
\newblock URL: \url{https://github.com/python/cpython/issues/78742}.

\bibitem{Wild2016}
Sebastian Wild.
\newblock {\em Dual-Pivot Quicksort and Beyond: Analysis of Multiway
  Partitioning and Its Practical Potential}.
\newblock Dissertation ({Ph.\,D.\ }thesis), 2016.
\newblock URL: \url{https://www.wild-inter.net/publications/wild-2016}.

\end{thebibliography}

\clearpage
\onecolumn
\appendix
\ifkoma{\addpart{Appendix}}{}

\section{C++ Code}
\label{app:code}

We provide some key parts of our \Cpp implementation
(slightly redacted for readability).

\subsection{4-way merge with sentinels} 
We first give the sentinel-based 4-way merge method; this is the method that gave the best performance.

\begin{lstlisting}[language=C++,gobble=4]
    /**
     * Merge runs [l..g1) and [g1..g2) and [g2..g3) and [g3..r) in-place into [l..r)
     * using a buffer at B of length at least r-l+4.
     */.
    template<typename Iter, typename Iter2>
    void merge_4runs(Iter l, Iter g1, Iter g2, Iter g3, Iter r, Iter2 B) {
        auto n = r - l;
        // Copy runs to B and append a sentinel value after each.
        std::copy(l,  g1, B);                 *(B + (g1 - l))     = plus_inf_sentinel();
        std::copy(g1, g2, B + (g1 - l) + 1);  *(B + (g2 - l) + 1) = plus_inf_sentinel();
        std::copy(g2, g3, B + (g2 - l) + 2);  *(B + (g3 - l) + 2) = plus_inf_sentinel();
        std::copy(g3, r,  B + (g3 - l) + 3);  *(B + (r  - l) + 3) = plus_inf_sentinel();
        Iter2 c[4];  // pointers to runs in B.
        c[0] = B, c[1] = B + (g1-l)+1, c[2] = B + (g2-l)+2, c[3] = B + (g3-l)+3;
        // initialize tournament tree
        Iter2 x, y;  std::pair<Iter2, bool> z;
        //       z                x,y,z store iterator                      
        //    /     \             from winning run;         
        //   x       y            z also whether min         
        //  / \     / \           min came from left subtree           
        // 0   1   2   3                      
        if (*c[0] <= *c[1]) x = c[0]++; else x = c[1]++;
        if (*c[2] <= *c[3]) y = c[2]++; else y = c[3]++;
        if (*x <= *y) z = {x, true}; else z = {y, false};
        
        *l++ = *(z.first); // vacate root into output
        for (auto i = 1; i < n; ++i) {
            if (z.second) { // min came from c[0] or c[1], so recompute x.
                if (*c[0] <= *c[1]) x = c[0]++; else x = c[1]++;
            } else { // min came from c[2] or c[3], so recompute y.
                if (*c[2] <= *c[3]) y = c[2]++; else y = c[3]++;
            }
            // always recompute z
            if (*x <= *y) z = {x, true}; else z = {y, false};
            *l++ = *(z.first);
        }
    }
\end{lstlisting}

\subsection{3-way merge and 2-way merge}

The method for merging 3 runs is similar, but doesn't use the right subtree;
instead \texttt{y} is simply also set to \texttt{c[2]}.
For 2-way merging many variations have been explored;
when the type supports a sentinel value, the following method has a very efficient merging loop.

\begin{lstlisting}[language=C++,gobble=4]
    void merge_2runs(Iter l, Iter m, Iter r, Iter2 B) {
        auto n1 = m-l, n2 = r-m;
        std::copy(l, m, B);          *(B+(m-l))   = plus_inf_sentinel();
        std::copy(m, r, B+(m-l+1));  *(B+(r-l)+1) = plus_inf_sentinel();
        auto c1 = B,  c2 = B + (m - l + 1),  o = l;
        while (o < r) *o++ = *c1 <= *c2 ? *c1++ : *c2++;
	}
\end{lstlisting}

\subsection{4-way Powersort}

Our implementation of 4-way Powersort uses a stack of records maintained in a fixed-size array.
As in Timsort, we extend the natural runs in the input to a minimal length \texttt{minRunLen};
setting this parameter to 1 disables that optimization.

\begin{lstlisting}[language=C++,gobble=8]
        struct run {  Iter begin; Iter end;  };
        struct run_n_power {  Iter begin; Iter end; int power = 0;  };
        struct run_begin_n_power{  Iter begin;  int power;  };
        run_begin_n_power NULL_RUN_N_POWER{};
		
        /** sorts [begin,end) */
        void powersort_4way(Iter begin, Iter end) {
            auto n = end - begin;
            auto maxStackHeight = 3*(ceil_log4(n)+1);
            auto stack = new run_n_power[maxStackHeight];
            run_begin_n_power *top = stack; // topmost valid stack element
            *top = NULL_RUN_N_POWER; // power 0 as sentinel entry
            run_begin_n_power * const endOfStack = stack + maxStackHeight;

            run_n_power runA = {begin, extend_and_reverse_run(begin, end), 0};
            // extend to minimal run length
            if (auto lenA = runA.end - runA.begin < minRunLen) {
                runA.end = std::min(end, runA.begin + minRunLen);
                insertionsort(runA.begin, runA.end, lenA); // skips first lenA iterations
            }
            while (runA.end < end) {
                run runB = {runA.end, extend_and_reverse_run(runA.end, end)};
                if (size_t lenB = runB.end - runB.begin < minRunLen) {
                    runB.end = std::min(end, runB.begin + minRunLen);
                    insertionsort(runB.begin, runB.end, lenB);
                }
                runA.power = node_power(0, n, runA.begin - begin, runB.begin - begin, runB.end - begin);
                // Invariant: powers on stack are weakly increasing from bottom to top
                while (top->power > runA.power) 
                    merge_loop(top, runA);
                *(++top) = {runA.begin, runA.power}; // push new run onto stack
                runA = {runB.begin, runB.end, 0};
            }
            merge_down(stack, top, runA);
            delete[] stack;
        }
\end{lstlisting}

To make the code more readable, the body of the main loop has been moved to a function \texttt{merge\_loop}; \texttt{g++} inlined the call as part of its compiler optimizations.
\texttt{merge\_loop} counts the number of (contiguous) entries on top of the stack with equal power and then merges these with \texttt{runA}.

\begin{lstlisting}[language=C++,gobble=8]
        void merge_loop(run_begin_n_power * &top_of_stack, run_n_power &runA) {
            int nRunsSamePower = 1;
            while((top_of_stack - nRunsSamePower)->power == top_of_stack->power)
                ++nRunsSamePower;
            if (nRunsSamePower == 1) { // 2way
                Iter g[] = {top_of_stack->begin};
                merge_2runs(g[0], runA.begin, runA.end, _buffer.begin());
                runA.begin = g[0];
            } else if (nRunsSamePower == 2) { // 3way
                Iter g[] = {(top_of_stack-1)->begin, top_of_stack->begin};
                merge_3runs(g[0], g[1], runA.begin, runA.end, _buffer.begin());
                runA.begin = g[0];
            } else { // 4way
                Iter g[] = {(top_of_stack-2)->begin, (top_of_stack-1)->begin, top_of_stack->begin};
                merge_4runs(g[0], g[1], g[2], runA.begin, runA.end, _buffer.begin());
                runA.begin = g[0];
            }
            top_of_stack -= nRunsSamePower; // pop runs
        }
\end{lstlisting}

The function \texttt{merge\_down} successively merges the top 4 elements on the stack;
since runs are typically exponentially increasing in size as we work our way through the stack,
it is beneficial to first bring the number of runs to $3k+1$ for a $k\in\N$ using a single 2-way or 3-way merge at the beginning. Then all remaining merges are 4-way merges.

\begin{lstlisting}[language=C++,gobble=8]
        void merge_down(run_begin_n_power *begin_of_stack, run_begin_n_power * &top_of_stack, run_n_power &runA) {
            // we have the entire stack of runs, so instead of following exactly the powersort rule, we can
            // be slightly more clever and make sure we have 4way merges all the way through except the first merge
            auto nRuns = top_of_stack - begin_of_stack + 1; // stack size + runA
            // We want 3k+1 runs, so that repeatedly merging 4 and putting the result back gives 4way merges all the way through.
            switch (nRuns % 3) {
                case 0: // merge topmost 3 runs
                    merge_3runs<mergingMethod>((top_of_stack-1)->begin, top_of_stack->begin,
                                               runA.begin, runA.end, _buffer.begin());
                    runA.begin = (top_of_stack-1)->begin;
                    top_of_stack -= 2;
                    break;
                case 2: // merge topmost 2 runs
                    merge_2runs(top_of_stack->begin, runA.begin, runA.end, _buffer.begin());
                    runA.begin = top_of_stack->begin;
                    --top_of_stack;
                    break;
                default:
                    break;
            }
            assert(((top_of_stack - begin_of_stack) % 3) == 0);
            // merge remaining stack 4way each
            while (top_of_stack > begin_of_stack) {
                merge_4runs((top_of_stack-2)->begin, (top_of_stack-1)->begin,
                                                 top_of_stack->begin, runA.begin, runA.end, _buffer.begin());
                runA.begin = (top_of_stack-2)->begin;
                top_of_stack -= 3;
            }
        }
\end{lstlisting}

Note that this merge-down strategy can lead to a slightly different merge tree compared to the ``pure Powersort'' merge tree realized by \texttt{kway\_tree}$_k$; however, it is easy to show this change can only reduce the resulting merge cost.

\subsection{Sentinel-free 2-way merge}

If $+\infty$ values are not available, we need to include boundary checks.
This changes the code to the following method.

\begin{lstlisting}[language=C++,gobble=4]
	void merge_2runs_no_sentinel(Iter l, Iter m, Iter r, Iter2 B) {
		auto n1 = m-l, n2 = r-m;
		std::copy(l,r,B);
		auto c1 = B, e1 = B + n1, c2 = e1, e2 = e1 + n2;  auto o = l;
		while (c1 < e1 && c2 < e2)  *o++ = *c1 <= *c2 ? *c1++ : *c2++;
		while (c1 < e1)  *o++ = *c1++;
		while (c2 < e2)  *o++ = *c2++;
	}
\end{lstlisting}

\noindent
For 2-way merging, we also use the following merge method that only copies the smaller run to the buffer.

\begin{lstlisting}[language=C++,gobble=4]
	void merge_2runs_copy_smaller(Iter l, Iter m, Iter r, Iter2 B) {
		auto n1 = m-l, n2 = r-m;
        if (n1 <= n2) {
            std::copy(l,m,B);
            auto c1 = B, e1 = B + n1;  auto c2 = m, e2 = r, o = l;
            while (c1 < e1 && c2 < e2)  *o++ = *c1 <= *c2 ? *c1++ : *c2++;
            while (c1 < e1) *o++ = *c1++;
        } else {
            std::copy(m,r,B);
            auto c1 = m-1, s1 = l, o = r-1;  auto c2 = B+n2-1, s2 = B;
            while (c1 >= s1 && c2 >= s2)  *o-- = *c1 <= *c2 ? *c2-- : *c1--;
            while (c2 >= s2) *o-- = *c2--;
        }
	}
\end{lstlisting}

\subsection{Sentinel-free 4-way merge (merging by stages)}

We created a 4-way merging method that does not use sentinels and works in ``stages''
(as described in the main text).
We use \Cpp templates to exploit similarities of the stages;
during compilation, these templates are unfolded and code optimization runs on the deflated code.

\begin{lstlisting}[language=C++,gobble=4]
    struct tournament_tree_node {  Iter2 it; bool fromRun0Or1;  };
    enum number_runs {  TWO = 2, THREE = 3, FOUR = 4  };
    
    template<number_runs nRuns> 
    long compute_safe(std::vector<Iter2> & c, std::vector<Iter2> & e, std::vector<long> & nn) {
        for (int i = 0; i < nRuns; ++i) nn[i] = e[i] - c[i];
        long safe = *(std::min_element(nn.begin(), nn.end()));
        return safe;
    }
    
    template<number_runs nRuns>
    void initialize_tournament_tree(std::vector<Iter2> &c, std::vector<Iter2> &e,
                                    std::array<tournament_tree_node<Iter2>, 3> &N) {
        // tourament tree:
        //      N[0]
        //    /     \
        //  N[1]    N[2]
        //  / \     / \
        // 0   1   2   3
        if (*c[0] <= *c[1]) N[1] = {c[0]++, true}; else N[1] = {c[1]++, true};
        if (nRuns == 4)
            if (*c[2] <= *c[3]) N[2] = {c[2]++, false}; else N[2] = {c[3]++, false};
        else
            N[2] = {c[2]++, false};
        N[0] = *(N[1].it) <= *(N[2].it) ? N[1] : N[2];
    }
    
    template<number_runs nRuns>
    void update_tournament_tree(std::vector<Iter2> &c, std::vector<Iter2> &e,
                                std::array<tournament_tree_node<Iter2>, 3> &N) {
        if (N[0].fromRun0Or1) {
            if (*c[0] <= *c[1]) N[1] = {c[0]++, true}; else N[1] = {c[1]++, true};
        } else { // otherwise min came from c[2] or c[3], so recompute y.
            if (nRuns == 4)
                if (*c[2] <= *c[3]) N[2] = {c[2]++, false}; else N[2] = {c[3]++, false};
            else
                N[2] = {c[2]++, false};
        }
        // always recompute z
        N[0] = *(N[1].it) <= *(N[2].it) ? N[1] : N[2];
    }
    
    template<number_runs nRuns>
    bool rollback_tournament_tree(std::vector<Iter2> &c, std::vector<Iter2> &e,
                                  std::array<tournament_tree_node<Iter2>, 3> &N,
                                  std::vector<long> &nn) {
        auto other = N[0].fromRun0Or1 ? N[2] : N[1];
        // roll back into 'its' run
        int rollbacks = 0;
        for (auto i = 0; i < nRuns; ++i)
            if (c[i] - 1 == other.it) {
                --c[i], ++nn[i], ++rollbacks; break;
            }
        int i = std::find(nn.begin(), nn.end(), 0) - nn.begin();
        if (i == nRuns) {
            // rolled back into run that got empty; nasty special case.
            // But we made progress in the root, so just continue one more round with same nRuns.
            // need to rebuild the tree for that
            initialize_tournament_tree<nRuns>(c, e, N);
            return false;
        } else {
            c.erase(c.begin() + i);
            e.erase(e.begin() + i);
            return true;
        }
    }
    
    template<number_runs nRuns>
    bool do_merge_runs(Iter & l, Iter const r, std::vector<Iter2> &c, std::vector<Iter2> &e) {
        if (nRuns == TWO) {
            // simple twoway merge
            while (c[0] < e[0] && c[1] < e[1])   *l++ = *c[0] <= *c[1] ? *c[0]++ : *c[1]++;
            while (c[0] < e[0]) *l++ = *c[0]++;
            while (c[1] < e[1]) *l++ = *c[1]++;
            return true;
        } else {
            // use tournament tree
            std::array<tournament_tree_node<Iter2>, 3> N;
            initialize_tournament_tree<nRuns>(c, e, N);
            std::vector<long> nn(nRuns); // run sizes
            while (l < r) {
                long safe = compute_safe<nRuns>(c, e, nn);
                if (safe > 0) {
                    for (; safe > 0; --safe) {
                        *l++ = *(N[0].it); // output root
                        update_tournament_tree<Iter2, nRuns>(c, e, N);
                    }
                } else {
                    // one run is exhausted; need to handle elements in the tree
                    *l++ = *(N[0].it); // easy for the root (guaranteed min)
                    // rollback other element into its run
                    if (rollback_tournament_tree<nRuns>(c, e, N, nn))
                        // occasionally, we rollback into an empty run and have to keep going; otherwise, terminate loop.
                        break;
                }
            }
            return false;
        }
    }
       
    void detect_and_remove_empty_runs(std::vector<Iter2> & c, std::vector<Iter2> & e) {
        int nRuns = c.size();  long safe;
        while (true) {
            std::vector<long> nn(nRuns);
            for (int i = 0; i < nRuns; ++i) nn[i] = e[i] - c[i];
            safe = *(std::min_element(nn.begin(), nn.end()));
            if (safe > 0) return;
            int i = std::find(nn.begin(), nn.end(), 0) - nn.begin();
            c.erase(c.begin() + i);  e.erase(e.begin() + i);  --nRuns;
        }
    }
    
    template<typename Iter, typename Iter2>
    void merge_4runs_no_sentinels(Iter l0, Iter g1, Iter g2, Iter g3, Iter r, Iter2 B) {
        Iter l = l0;  const auto n = r - l;
        std::copy(l,  g1, B);
        std::copy(g1, g2, B + (g1 - l));
        std::copy(g2, g3, B + (g2 - l));
        std::copy(g3,  r, B + (g3 - l));
        *(B+n) = *(B+n-1); // sentinel value so that accesses to endpoints don't fail
        std::vector<Iter2> c {B, B+(g1-l), B+(g2-l), B+(g3-l)     }; // current element
        std::vector<Iter2> e {   B+(g1-l), B+(g2-l), B+(g3-l), B+n}; // endpoints
        detect_and_remove_empty_runs(c, e);
        while (l < r) {
            switch (c.size()) {
                case 4:  if (do_merge_runs<FOUR> (l, r, c, e)) break;
                case 3:  if (do_merge_runs<THREE>(l, r, c, e)) break;
                case 2:  if (do_merge_runs<TWO>  (l, r, c, e)) break;
                case 1:  return;
            };
        }
    }
\end{lstlisting}

\ifdraft{
	\clearpage
	\part*{Notes-to-self}
	\printnotestoself
}{}

\end{document}